\documentclass[11pt]{elsarticle}

\usepackage{amsmath,amsfonts,amssymb,amsthm,epsfig,epstopdf,url,array}
\usepackage{multicol}
\usepackage{caption}
\usepackage{subcaption}
\usepackage{epsfig,color}
\usepackage[english]{babel}
\usepackage{mleftright}
\usepackage{graphicx}
\usepackage{bm}
\usepackage{array}
\usepackage{datetime}
\usepackage{bbold}
\usepackage{amsfonts}
\usepackage{relsize}
\usepackage{tikz-cd}
\usepackage{fancyhdr}
\usepackage{graphicx}
\usepackage{mathtools}
\usepackage{tikz}
\usepackage{comment}
\usetikzlibrary{fit,arrows,calc,positioning}
\tikzset{
block/.style={
  draw,
  rectangle,
  minimum height=1.5cm,
  minimum width=3cm, align=center
  },
line/.style={<-,>=latex'}
}

\theoremstyle{plain}
\newtheorem{thm}{Theorem}[section]
\newtheorem{theorem}{Theorem}[section]

\newtheorem{lemma}[thm]{Lemma}
\newtheorem{prop}[thm]{Proposition}

\newtheorem{corollary}{Corollary}

\theoremstyle{definition}

\newtheorem{definition}{Definition}[section]

\newtheorem{rem}{\textbf{Remark}}

\newcommand{\dsp}{\displaystyle}
\newcommand{\dd}{\text{d}}

\usepackage{vector}

\usepackage{amssymb}

\usepackage{lineno}
\usepackage[hidelinks]{hyperref}
\modulolinenumbers[5]

\journal{Journal of \LaTeX\ Templates}

\begin{document}
\begin{frontmatter}

\title{
Timoshenko beam under finite and dynamic  transformations: Lagrangian coordinates and Hamiltonian structures}
\author[oscar]{Oscar Cosserat \fnref{oc}}
\author[loic]{Lo\"{\i}c  Le Marrec\fnref{ll}}
\address[oscar]{Georg-August-University G\"ottingen, Bunsenstra{\ss}e 3-5, 37073, G\"ottingen, Germany}
\address[loic]{Univ Rennes, CNRS, IRMAR - UMR 6625, F-35000 Rennes, France}
\fntext[oc]{oscar.cosserat@mathematik.uni-goettingen.de}
\fntext[ll]{loic.lemarrec@univ-rennes.fr}

\begin{abstract}
In the framework of Timoshenko beam, the material parameters are inherently prescribed on the material moving frame. In this regard, we derive the strong and weak formulations of the dynamics under finite transformation in Lagrangian coordinates. Accordingly, analytical mechanics tools are used to deduce a new Hamiltonian formulation of the model which proves to be remarkably simple and synthetic.
\end{abstract}
\end{frontmatter}

\section*{Acknowledgments and data statement}
We wish to thank all attendees of the GDR CNRS $n^\circ 2043$ "G\'eom\'etrie Diff\'erentielle et M\'ecanique" which, through numerous exchanges, helped us to clear out our vision of the different topics related to this work. The first author acknowledges the Deutsche Forschungsgemeinschaft RTG 2491 for its financial support. This manuscript has no associated data.

\tableofcontents

\newpage

\section{Introduction}
Analysis of one-dimensional structures began in the 15th century, thanks to the pioneering work of Leonardo da Vinci. Later on, Leonhard Euler, James and Daniel Bernoulli provided detailed explanations and analyses regarding one-dimensional structures, specifically focusing on static of beams \cite{euler1744methodus,BookLevien}. During the 20th century, Stephen Timoshenko further extended the understanding of one-dimensional structures \cite{elishakoff2020developed}. Inspired by three-dimensional elasticity and solutions of the Saint-Venant problem, he provided an alternative of Euler-Bernoulli beam theory. In practice, Timoshenko's work involves the effects of shear deformation and rotational inertia, which were not accounted in the original theory \cite{Timoshenko1921lxvi}. By incorporating these additional factors, Timoshenko beam theory provided a more comprehensive and accurate representation of beam behavior \cite{simao2017influence}. \\
This more general model has motivated numerous studies on various aspects of structural mechanics. These include investigations into elasticity \cite{sankar2001elasticity,rakotomanana2006elements,le2018vibration}, plasticity and damage \cite{bavzant2002nonlocal} or buckling \cite{wang2006buckling,spagnuolo2019targeted}. Even if exact solutions are available for statics and dynamics in a linear regime, solutions of the non-linear dynamical approach were mainly given using numerical methods due to the difficulty of solving analytically such non-linear partial differential equations. Some analytical solutions of the static non-linear problem are available in the planar case (\textit{e.g.} \cite{MarwanPoutre}) and give, in turn, interesting consequences on stability \cite{koutsogiannakis2023double} or buckling \cite{hariz2022buckling}. \\
In the second part of the twentieth century the equations governing the dynamics under large transformation of such Timoshenko beam have been obtained in a geometrically exact way through the pioneer work of Reissner \cite{reissner1973one,reissner1972one}, Whithman \cite{whitman1974exact} and Simo \cite{simo1985finite}. A modern synthesis can be found in \textit{e.g.} \cite{mata2008inelastic,attard2003finite}. In practice, these kinematical formulations are based on the analysis of a one-dimensional Cosserat structure embedded in the Euclidean space \cite{cosserat1909theorie}. Taking into account the kinematics of the rigid frame associated to the section-plane allows to avoid kinematical constraints inherent to Euler-Bernoulli beam. Interestingly, these additional degrees of freedom avoid Lagrange multipliers on the associated variational principle.\\
Mathematical modeling of the geometrically-exact beam model has been an active field of research for decades now. The derivation of the model is revisited in \cite{kapania2003} with emphasis on initially curved beam in order to compute static problems \cite{kapania2003_2}. In \cite{auricchio2008}, a splitting of the deformation gradient is used to compute the model from three-dimensional finite-elasticity and study the finite-deformation small-strain case. \cite{omurtag2016} applies the Frenet–Serret frame to ensure strain measure objectivity in a finite element implementation. Modeling for finite element approach can also be found in \cite{gruttmann2000,levyakov2015formulation}, and in the series of articles \cite{pimenta2019, pimenta2020, pimenta2021} using Rodrigues formula. An interested reader might consult \cite{murakami2017_1} and \cite{murakami2017_2} for a Cartan frame perspective, where a more differential geometry oriented framework is applied to the analytical mechanics of beam theory.\\
From an other point of view, it is well known that considering mechanics with respect to Hamiltonian formulation provides interesting insights (see for instance \cite{lanczos2012variational}). It leads to study Lie algebras and invariant theory \cite{marsden2013introduction}. On this aspect, the Hamiltonian structure of Timoshenko beam under large transformation has been less analyzed. One of the reason is that the geometrically exact formulation is quite recent \cite{chadha2017introductory}. An other reason is that for the study of finite transformations -- in opposition to the linearized problem -- various formulations are possible and generate some confusion between the tangent space being convected by the director-frame and the standard tangent space of the Euclidean space. As an example, frame-invariance is formulated in a different way depending on the choice of a material or spatial point of view \cite{IBRAHIMBEGOVIC20025159}. \\
In their seminal work, Simo and Vu-Quoc (see \cite{simo1989}, rk 6.1) noticed the existence of one Hamiltonian formulation through a variational principle and its usefulness for computational purposes. In \cite{simo1988hamiltonian}, a rigorous construction of Hamilton equations and appropriate Poisson brackets is provided. There, this construction is aimed at studying conservation laws and Noether theorem on one hand, and stability and bifurcations on the other hand \cite{cedolin2010stability}. A crucial use is made of left multiplication by rotations in order to deal with intrinsic group invariance. A Hamiltonian formulation of geometrically exact beam is provided in \cite{chadha2021poisson} using Poisson brackets through a Legendre transformation. However, none of them  fully rely on Lagrangian coordinates. \\
In turn, geometrical studies of Timoshenko model apply to numerical purposes. Simo \emph{et al.} \cite{simo1995} studied the geometric structure of the three-dimensional non-linear dynamics of rods in order to develop and benchmark geometric numerical methods. This aimed at comparing this geometric approach to more traditional numerics. Simo's formulation is already applied in \cite{crisfield1999objectivity} to construct an objective finite element formulation. Ghosh and Roy \cite{ghosh2008, ghosh2009} investigated parametrization of rotations for objectivity--ensuring numerical schemes, discussing whether quaternions or rotation vectors should be used. Demoures \emph{et al.} \cite{demoures2015discrete}, Hante and co-authors \cite{hante2022lie} and Herrmann and Kotyczka \cite{herrmann2024} derived structure preserving integrators for geometrically exact beam dynamics by using a Lie group variational integrator. The latter even investigated a numerical framework to handle dissipation. About ten years ago, \cite{zhong2016} reformulated the objectivity property of a given finite element formulation by the equivariance of the corresponding interpolation operator.\\
These last decades, the study of symmetries of Timoshenko model attracted more and more attention. Here, too, Lie group structures can be used in beam theory to describe the invariances and symmetries associated with beam motions and deformations. In \cite{MarwanLie,lerbet2005intrinsic,lerbet2003general}, a theoretical study of one-dimensional Cosserat structures has been performed using the Lie algebra of the displacement group. Djondjorov \cite{djondjorov1995invariant} studied invariance properties of the Timoshenko beam equations using a Lie group formalism. \\
A point has been neglected so far. In order to preserve objectivity and material invariance, the parameters of the model are naturally expressed in material coordinates. Therefore, it is natural to study the Timoshenko model in Lagrangian coordinates: by avoiding any decomposition of tensors on the Cartesian frame associated to the ambient space.  This is the main motivation of the present paper. \\
\\
\emph{Outline of the paper} Basics on beam modeling are first presented in an extensive way, in section \ref{Chapitre1}, by describing in a detailed manner all the mathematical tools used to describe such Cosserat structure involving both first-order (\textit{e.g.} translation) and second order (\textit{e.g.} rotation) tensors. Particular attention is paid to the choice of formulation, privileging either components or full tensors, in the moving frame convected by the beam section. In the section \ref{Chapitre3}, the main mechanical objects are introduced: strains, stress, energy densities, etc. This allows to expose in a clear way conservation equation, compatibility conditions and a variational formulation of the dynamical problem. Again, the choice of formulation is discussed: mechanical quantities are expressed on a moving frame whereas numerical resolution invokes algebraic formulations. Once these mechanical preliminaries are stated, the Hamiltonian formalism is presented and discussed in section \ref{sec:ham_mechanics_tim}. We devote one part of this section to a detailed explanation of the choice of variables. Using the mobile frame, we exhibit a set of variables that is different from the one we used for the equilibrium equations. In addition to the Poisson brackets already mentioned in the literature \cite{simo1988hamiltonian}, two  Poisson brackets are constructed to recover equations of motion through the Hamilton equation. The last is new and is proved to deliver time-variations of curvature $\boldsymbol{\kappa}$ and strain-vector $\boldsymbol{\varepsilon}$. 
\section{Cosserat beam formalism and associated tools}\label{Chapitre1}
One considers a one-dimensional Cosserat beam model \cite{cosserat1909theorie} defined by a material curve $\mathfrak{C}$ called the fiber. $\mathfrak{C}$ lies in a three-dimensional, and oriented, Euclidean space $\mathbb{R}^3$. In practice, the curve corresponds to the position of the centroids $G$ of the cross-sections $\mathfrak{S}$ of the beam. In the stress-free reference configuration, the fiber $\mathfrak{C}$ is a straight segment of length $L$ and the sections are normal to the spatial curve. These sections are uniform all along the fiber. For such a beam model, the sections $\mathfrak{S}$ of the beam are supposed to be rigid whatever the transformation. Hence, placement of centroid and orientation of the cross-section are two independent degrees of freedom. In particular, the sections are not necessarily normal to the fiber $\mathfrak{C}$ as opposed to Euler-Bernoulli beam assumption. We explain below the mathematical counterpart of these assumptions.

\subsection{Modeling}
Considering $S \in [0, L]$ as a material curvilinear coordinate of $\mathfrak{C}$, the placement of $G$ after transformation is defined by the map 
$$
 \dsp S \in [0, L]  \rightarrow \boldsymbol{\varphi}(S,t) := \bvec{OG}(S,t)\ \in \mathbb{R}^3.
$$
where $O$ is an arbitrary fixed origin of the surrounding space. \\
For such a Cosserat-like structure it is justified to use a moving orthonormal frame, called the director frame basis $\{\bvec{d_i}\}:=(\bvec{d_{1}},\bvec{d_{2}},\bvec{d_{3}})$, for which $(\bvec{d_{1}},\bvec{d_{2}})$ is a principal basis of the cross-section $\mathfrak{S}$ and $\bvec{d_{3}}$ is  normal to $\mathfrak{S}$. \\
In a general configuration, the basis $\{\bvec{d_i}\}$ depends on the time and the curvilinear abscissa $S$ of the beam, while, in the reference configuration, $\{\bvec{d_i}\}$ coincides with a chosen Cartesian frame $\{\bvec{e_I}\}:=(\bvec{e_1,e_2,e_3})$, assumed to be constant with respect to  $t$ and $S$. \\
Since $\{\bvec{e_{I}}\}$ and $\{\bvec{d_{i}}\}$ are both orthonormal basis  having the same orientation, a rotation tensor $\bvec{R}(S,t) \in SO(3)$ relates each basis
\begin{equation}\label{DEFrotation}
\bvec{d_i}(S,t)=\bvec{R}(S,t)\bvec{e_i}.
\end{equation}
A configuration is completely defined by the determination of $\boldsymbol{\varphi}(S,t)$ and $\bvec{R}(S,t)$. 
Note that $\{\bvec{d_i}\}$ is a material and convected basis, and $S$ is a material coordinate. Therefore, $\boldsymbol{\varphi}$ is not assumed to be an isometry onto its image: $\|\frac{\partial \boldsymbol{\varphi}}{\partial S}\| \neq 1$ a priori. This allows compressibility of the fiber $\mathfrak{C}$.

\subsection{Mobile frame notation}
Let us consider\footnote{Einstein notation is used all along the paper} a vector $\bvec{u}=u_i\bvec{d_{i}}$ expressed in the moving frame $\{\bvec{d_i}\}$. The components $u_i$ are said to be Lagrangian coordinates. An alternative notation is:
$$
u:=
\begin{pmatrix}
    u_{1}\\
    u_{2}\\
    u_{3}
\end{pmatrix}
\quad\quad\quad\textrm{and then}\quad\quad
d_{1}:=
\begin{pmatrix}
    1\\
    0\\
    0
\end{pmatrix}
, \quad\
d_{2}:=
\begin{pmatrix}
    0\\
    1\\
    0
\end{pmatrix}
, \quad\
d_{3}:=
\begin{pmatrix}
    0\\
    0\\
    1
\end{pmatrix}.
$$
Within this notation, the \emph{moving} aspect of the frame is lost. We highlight these differences by the following:
\begin{itemize}
    \item vectors with \textit{algebraic} notation are denoted by a regular font, \textit{e.g.} $u$
    \item vectors \textit{in the moving frame} are denoted by a fat font, \textit{e.g.}  $\bvec{u}=u_i\bvec{d_{i}}.$ We call this the \textit{mobile} notation.
\end{itemize}
These requires particular attention when one comes to derivations. We will highlight this in section \ref{sec:derivation}.

\subsection{Properties of the cross product}
To fix ideas, let us introduce some useful mathematical tools. For vectors in $\mathbb{R}^3$, $\wedge$ and $<,>$ are used for the cross product and for the Euclidean scalar product, respectively. The commutator of 2 matrices $\bvec{A}$ and $\bvec{B}$ is set to be $[\bvec{A},\bvec{B}]=\bvec{A}\bvec{B}-\bvec{B}\bvec{A}$ whereas the scalar product of matrices is defined by:
\begin{equation}\label{eq:defmetric}            
\ll \bvec{A},\bvec{B}\gg  := \frac{1}{2}\text{Tr}(\bvec{A}^{T} \cdot \bvec{B}).
        \end{equation}
Let us denote $j$ the map
\begin{equation}
    \begin{array}{ccccc}
        j &\colon & so(3) &\to &\mathbb{R}^3  \\
           && \begin{pmatrix}
               0 & -u_3 & u_2 \\
               u_3 & 0 & -u_1 \\
               -u_2 & u_1 & 0
           \end{pmatrix}& \mapsto& \begin{pmatrix}
               u_1\\
               u_2\\
               u_3
           \end{pmatrix}.
    \end{array}
\end{equation}

We will use the following properties of the map $j$ all along this article.
\begin{prop}\label{prop:j}
    The map $j$ is a linear isomorphism. Furthermore,
    \begin{enumerate}
        \item it is an isometry: for any $\bvec{A}$ and $\bvec{B}$ in $so(3),$
        \begin{equation}
            <j(\bvec{A}),j(\bvec{B})> = \ll \bvec{A},\bvec{B}\gg
        \end{equation}
        \item it is a Lie morphism:
        \begin{equation}
            [\bvec{A},\bvec{B}] = j(\bvec{A}) \wedge j(\bvec{B})
        \end{equation}
        \item for any $\bvec{u} \in \mathbb{R}^3,$
            \begin{equation}
                \bvec{A} \cdot \bvec{u} = j(\bvec{A}) \wedge \bvec{u}.
            \end{equation}
    \end{enumerate}
\end{prop}
Note that the last expressions are base-dependent in the sense that cross product is considered on columnn vectors associated to a given vector basis. In the following, all computations are performed on the mobile frame $\{\bvec{d_i}\}$. Cartesian components are never invoked in these computations.

\subsection{Curvature and spin}
According to \eqref{DEFrotation}, spatial derivation of directors writes:
\begin{equation}
\label{directors}
    \frac{\partial \bvec{d_{i}}}{\partial S}=\mathbb{K}\bvec{d_{i}}
    \quad\quad\quad
    \textrm{where}
    \quad\quad\quad
    \mathbb{K}(S,t)=\bvec{R}^{-1}\frac{\partial \bvec{R}}{\partial S}
\end{equation}
By writing $\mathbb{K}(S,t),$ one highlights that this tensor is time-dependent and not homogeneous. As $\mathbb{K}(S,t) \in so(3)$, one introduces the curvature
$$
\boldsymbol{\kappa} :=j(\mathbb{K})
$$
and then \eqref{directors} can be formulated with the cross product:
\begin{equation}\label{deriveDIS}
    \frac{\partial \bvec{d_{i}}}{\partial S}= \boldsymbol{\kappa} \wedge  \bvec{d_{i}}
\end{equation}
The reader may notice that, introducing the components $\kappa_{i}=\boldsymbol{\kappa}\cdot\bvec{d_{i}}$ of $\boldsymbol{\kappa}$ in the moving frame:
\begin{equation}\label{eq:kappa}
\kappa:=
\begin{pmatrix}
\kappa_{1}\\
\kappa_{2}\\
\kappa_{3}
\end{pmatrix}
\quad\quad\quad\textrm{and}\quad\quad
\mathbb{K}:=
\begin{pmatrix}
0			&	-\kappa_{3}	&	\kappa_{2}\\
\kappa_{3}	&	0			&	-\kappa_{1}\\
-\kappa_{2}	&	\kappa_{1}	&	0\\
\end{pmatrix}.
\end{equation}
Equation \eqref{eq:kappa} gives the components of $\mathbb{K}$ in the \emph{moving frame}. \\
The same holds true for time derivation: introducing
\begin{equation}\label{timederivdirector}
\mathbb{W}(S,t):=\bvec{R}^{-1}\frac{\partial \bvec{R}}{\partial t },
\quad\quad\quad
\textrm{and the spin}
\quad\quad
\boldsymbol{\omega} :=j(\mathbb{W}),
\end{equation}
one obtains by differentiation with respect to the time
\begin{equation}\label{deriveDIt}
    \frac{\partial \bvec{d_{i}}}{\partial t}=\mathbb{W}\bvec{d_{i}}
    = \boldsymbol{\omega} \wedge  \bvec{d_{i}}
\end{equation}
and associated components in the moving frame are:
\begin{equation}
\omega:=
\begin{pmatrix}
\omega_{1}\\
\omega_{2}\\
\omega_{3}
\end{pmatrix}
\quad\quad\quad\textrm{and}\quad\quad
\mathbb{W}:=
\begin{pmatrix}
0			&	-\omega_{3}	&	\omega_{2}\\
\omega_{3}	&	0			&	-\omega_{1}\\
-\omega_{2}	&	\omega_{1}	&	0\\
\end{pmatrix}.
\end{equation}

\subsection{Derivations in material coordinates}\label{sec:derivation}
For any vector $\bvec{u}(S,t)$:
\begin{equation}\label{eq:time_der}
\frac{\partial \bvec{u}}{\partial S}=\frac{\partial u_i}{\partial S}\bvec{d_i}+\boldsymbol{\kappa}\wedge\bvec{u},
\quad \quad 
\frac{\partial \bvec{u}}{\partial t}=\frac{\partial u_i}{\partial t}\bvec{d_i}+\boldsymbol{\omega}\wedge\bvec{u},
\end{equation}
and
the Schwarz Lemma gives:
\begin{equation}\label{crossderivation1U}
\frac{\partial }{\partial t}\frac{\partial \bvec{u}}{\partial S} = \frac{\partial }{\partial S}\frac{\partial \bvec{u}}{\partial t} .
\end{equation}
We apply this to the material basis $\{\bvec{d_i}\}$ to obtain the following result.
\begin{prop}\label{prop:loic}
Time derivative of curvature and space derivative of spin are related by:
\begin{equation}\label{crossderivation2}
\frac{\partial \boldsymbol{\kappa}}{\partial t} = \frac{\partial \boldsymbol{\omega} }{\partial S} + \boldsymbol{\omega}\wedge \boldsymbol{\kappa}.
\end{equation}
\end{prop}
\begin{proof}
For any directors $\bvec{d_{i}}(S,t)$, the computation of crossed derivation gives:
\begin{align*}
\frac{\partial }{\partial t}\frac{\partial \bvec{d_{i}}}{\partial S} 
&=
\frac{\partial }{\partial t}\left(\boldsymbol{\kappa}\wedge \bvec{d_{i}} \right) 
&
\frac{\partial }{\partial S}\frac{\partial \bvec{d_{i}}}{\partial t} 
&=
\frac{\partial }{\partial S}\left(\boldsymbol{\omega}\wedge \bvec{d_{i}} \right) \\
&=
\frac{\partial \boldsymbol{\kappa} }{\partial t}\wedge \bvec{d_{i}} 
+
 \boldsymbol{\kappa} \wedge \frac{\partial \bvec{d_{i}}}{\partial t}	
 &
&=
\frac{\partial \boldsymbol{\omega} }{\partial S}\wedge \bvec{d_{i}} 
+
 \boldsymbol{\omega} \wedge \frac{\partial \bvec{d_{i}}}{\partial S}	\\
&=
\frac{\partial \boldsymbol{\kappa} }{\partial t}\wedge \bvec{d_{i}} 
+
 \boldsymbol{\kappa} \wedge \left(\boldsymbol{\omega}\wedge \bvec{d_{i}}\right)	
&
 &=
\frac{\partial \boldsymbol{\omega} }{\partial S}\wedge \bvec{d_{i}} 
+
 \boldsymbol{\omega} \wedge \left(\boldsymbol{\kappa}\wedge \bvec{d_{i}}\right)	
\end{align*}
Equating the two terms and using Jacobi identity leads to
\begin{align}
\begin{split}
\frac{\partial \boldsymbol{\kappa} }{\partial t}\wedge \bvec{d_{i}} 
+
 \boldsymbol{\kappa} \wedge \left(\boldsymbol{\omega}\wedge \bvec{d_{i}}\right)	
&=
\frac{\partial \boldsymbol{\omega} }{\partial S}\wedge \bvec{d_{i}} 
+
 \boldsymbol{\omega} \wedge \left(\boldsymbol{\kappa}\wedge \bvec{d_{i}}\right)	\\
 \frac{\partial \boldsymbol{\kappa} }{\partial t}\wedge \bvec{d_{i}} 
&=
\frac{\partial \boldsymbol{\omega} }{\partial S}\wedge \bvec{d_{i}} 
+
 \boldsymbol{\omega} \wedge \left(\boldsymbol{\kappa}\wedge \bvec{d_{i}}\right)	
+
 \boldsymbol{\kappa} \wedge \left(\bvec{d_{i}}\wedge \boldsymbol{\omega}\right)\\	
  \frac{\partial \boldsymbol{\kappa} }{\partial t}\wedge \bvec{d_{i}} 
&=
\frac{\partial \boldsymbol{\omega} }{\partial S}\wedge \bvec{d_{i}} 
+
 \left( \boldsymbol{\omega} \wedge \boldsymbol{\kappa} \right)	\wedge \bvec{d_{i}}
 \end{split}
\end{align}
for any $\bvec{d_{i}}$. This gives the desired result.
\end{proof}

\subsection{Perturbations in material coordinates}\label{sec:perturbation}
In the present work, an extensive use of linear perturbation is made. In particular, it will be a crucial tool for the variational principle of section \ref{sec:VariationalPrinciple1}. Therefore, this section provides important relations related to linear perturbations. At this stage, we highlight that perturbation $\delta$ and derivation are independent, hence, for any vector $\bvec{u}:$
\begin{align}\label{eq:perturbationderivees}
    \delta \frac{\partial \bvec{u}}{\partial S}= \frac{\partial \delta \bvec{u}}{\partial S}
, \quad \quad 
\delta \frac{\partial \bvec{u}}{\partial t}= \frac{\partial \delta \bvec{u}}{\partial t}
\end{align}
and the same holds for any matrix $\bvec{A}.$ We state now the analogous of \eqref{crossderivation2} for linear perturbations.
\begin{lemma}\label{lem:marsden}
Let $\bvec{R} \colon I \to SO(3)$ and $\Theta = \bvec{R}^{-1} \delta \bvec{R} \in so(3).$ For $\mathbb{K}(S,t)$ and $\mathbb{W}(S,t)$ defined in \eqref{directors} and \eqref{timederivdirector} respectively, one gets:
    \begin{equation}\label{eq:cocycle}
        \delta \mathbb{W} = \frac{\partial\Theta}{\partial t} + [\Theta, \mathbb{W}].
    \end{equation}
And similarly:
\begin{equation}\label{eq:cocycle_space}
    \delta \mathbb{K} = \frac{\partial \Theta}{\partial S} + [\Theta, \mathbb{K}].
\end{equation}
\end{lemma}

\begin{proof}
We compute both: 
\[
\left\{\begin{array}{ccc}
    \dsp \delta \mathbb{W} &=& \dsp \bvec{R}^{-1} \delta (\frac{\partial \bvec{R}}{\partial t})  -  \bvec{R}^{-1} \delta \bvec{R} \bvec{R}^{-1} \frac{\partial \bvec{R}}{\partial t} \\[8pt]
    \dsp  \frac{\partial \Theta}{\partial t} &=& \dsp \bvec{R}^{-1} \frac{\partial \delta \bvec{R}}{\partial t}  - \bvec{R}^{-1} \frac{\partial \bvec{R}}{\partial t} \bvec{R}^{-1} \delta \bvec{R}  
\end{array}.\right.
\]
Equations \eqref{eq:perturbationderivees} provide $\frac{\partial \delta \bvec{R}}{\partial t } = \delta\frac{\partial  \bvec{R}}{\partial t } $ and we conclude at equation \eqref{eq:cocycle}. Equation \eqref{eq:cocycle_space} is proven the same way by replacing time with space derivatives. 
\end{proof}

\begin{corollary}
    Following the definitions $\boldsymbol{\omega} = j(\mathbb{W})$ and $\boldsymbol{\kappa} = j(\mathbb{K})$, one introduces $\delta \boldsymbol{\theta} = j( \Theta)$. A consequence of the previous lemma is:
\begin{equation}\label{eq:var_omega_kappa}
    \left\{
    \begin{array}{ccc}
        \delta \boldsymbol{\omega} &=&\dsp  \frac{\partial \delta \boldsymbol{\theta}}{\partial t} + \delta \boldsymbol{\theta} \wedge \boldsymbol{\omega}\\ [8pt]
    
        \delta \boldsymbol{\kappa} &=&\dsp  \frac{\partial \delta \boldsymbol{\theta}}{\partial S} + \delta \boldsymbol{\theta} \wedge \boldsymbol{\kappa}
    \end{array}
\right. .
\end{equation}
\end{corollary}

\begin{proof}
Since $j$ is linear, it commutes with space and time derivations:
\begin{equation}
\dsp \frac{\partial j(\mathbb{W})}{\partial S}  \dsp = j (\frac{\partial \mathbb{W}}{\partial S})
\quad \quad \quad \textrm{and}\quad \quad 
\dsp \frac{\partial j(\Theta )}{\partial t}  \dsp = j (\frac{\partial \Theta}{\partial t}).
\end{equation}
Since $j$ is a Lie morphism, it commutes with Lie brackets:
\begin{equation}
        j([\Theta, \mathbb{W}]) = j(\Theta) \wedge j(\mathbb{W})
        \quad \quad \quad \textrm{and}\quad \quad 
        j([\Theta, \mathbb{K}]) = j(\Theta) \wedge j(\mathbb{K})
\end{equation}
and this concludes the proof.
\end{proof}

\subsection{Derivation and perturbation of a quadratic form}\label{DerivationHermitianProduct}
Let us consider, for any vectors $\bvec{u}(S,t)$ and $\bvec{v}(S,t),$ the quadratic form:
\begin{equation}
f(\bvec{u},\bvec{v}):=\bvec{u}\mathbb{X}\bvec{v},
\end{equation}
where $\mathbb{X}=\mathbb{X}^T$ is symmetric and has time and space independent components when expressed in $\{\bvec{d_i}\}$. Thanks to equation \eqref{eq:time_der}, time and space derivations of such quadratic forms are expressed in the following: 
\begin{prop}
    Time and space derivations of the quadratic form $f$ are respectively related to corotational time and space derivations of associated vectors. In equations:
\begin{align}\label{DerivHermitianProd}
\begin{split}
\frac{\partial\, f}{\partial t}= \left(\frac{\partial \bvec{u}}{\partial t}-\boldsymbol{\omega}\wedge\bvec{u}\right)\mathbb{X}\bvec{v} +  \bvec{u}\mathbb{X}\left(\frac{\partial \bvec{v}}{\partial t}-\boldsymbol{\omega}\wedge\bvec{v}\right)    \\
\frac{\partial\, f}{\partial S}= \left(\frac{\partial \bvec{u}}{\partial S}-\boldsymbol{\kappa}\wedge\bvec{u}\right)\mathbb{X}\bvec{v} +  \bvec{u}\mathbb{X}\left(\frac{\partial \bvec{v}}{\partial S}-\boldsymbol{\kappa}\wedge\bvec{v}\right)    
\end{split}
\end{align}
\end{prop}
\begin{proof}
Of course for \emph{algebraic} notation:
\begin{equation}\label{eq:der_quad_vec}
\frac{\partial\, f}{\partial t}=\frac{\partial}{\partial t}\left( u\mathbb{X}v \right) = \frac{\partial\, u}{\partial t}\mathbb{X}v  + u\mathbb{X}\frac{\partial\, v}{\partial t}  
\end{equation}
but it cannot be transposed if the \emph{mobile} notation is used, namely if the frame is time-dependent. We use equation \eqref{eq:time_der} to obtain
\begin{equation}
\frac{\partial u }{\partial t} = \frac{\partial \bvec{u}}{\partial t} - \boldsymbol{\omega}\wedge\bvec{u}.   
\end{equation}
The previous relation \eqref{eq:der_quad_vec} gives directly the desired result for time derivation and a similar proof holds for space derivation.
\end{proof}
The following consequence is obtained through differential calculus. 
\begin{corollary}
    The infinitesimal perturbation of the quadratic form $f$ follows in the same way:
\begin{align}\label{PerturbHermitianProd}
\delta f= \left(\delta  \bvec{u}-\delta \boldsymbol{\theta}\wedge\bvec{u}\right)\mathbb{X}\bvec{v} +  \bvec{u}\mathbb{X}\left(\delta \bvec{v}-\boldsymbol{\theta}\wedge\bvec{v}\right)   
\end{align}
\end{corollary}
$\delta u :=\delta  \bvec{u}-\delta \boldsymbol{\theta}\wedge\bvec{u}$ is sometimes called the \emph{corotational perturbation}.

\section{Constitutive laws and dynamical behavior}\label{Chapitre3}
One presents the main mechanical quantities allowing to formulate the dynamics of Timoshenko beam. These latter are associated to material parameters (rigidity tensors, inertia tensors) being uniform and constant on the mobile frame. 

\subsection{Internal free energy density}
The strains are defined by two quantities: the spatial curvature $\boldsymbol{\kappa}$ and the strain-vector $\boldsymbol{\varepsilon}$ where
\begin{equation}\label{StrainDefinition}
\boldsymbol{\varepsilon}:=\frac{\partial \boldsymbol{\varphi}}{\partial S}-\bvec{d_{3}}.
\end{equation}
\begin{rem}
    Let us use the linear perturbation tool of the previous section to compute $\delta \boldsymbol{\varepsilon}:$
    \begin{align}
        \delta \boldsymbol{\varepsilon} &= \frac{\partial \delta\boldsymbol{\varphi}}{\partial S}-\delta\bvec{d_3}\\
        &= \frac{\partial \delta\boldsymbol{\varphi}}{\partial S}-\delta\boldsymbol{\theta}\wedge\bvec{d_3}.\label{eq:pertubation_eps}
    \end{align}
    This computation will be used later on.
\end{rem}
We introduce the strain-energy density $U(\boldsymbol{\varepsilon},\boldsymbol{\kappa}).
$ For linear stress-strain relations, it is quadratic:
\begin{equation}\label{PsiGeneral}
U(\boldsymbol{\varepsilon},\boldsymbol{\kappa}) := \frac{1}{2} \boldsymbol{\varepsilon}\mathbb{G}\boldsymbol{\varepsilon}+\frac{1}{2}\boldsymbol{\kappa}\mathbb{H}\boldsymbol{\kappa},
\end{equation}
where the rigidity tensors are diagonal matrices in the $\{\bvec{d_{i}}\}$-frame:
$$
\mathbb{G}:=
\begin{pmatrix}
    GA  &   0   &   0 \\
    0   &   GA  &   0 \\
    0   &   0   &   EA
\end{pmatrix},
\quad\quad\quad
\mathbb{H}:=
\begin{pmatrix}
    EI_{1}  &   0   &   0 \\
    0   &   EI_{2}  &   0 \\
    0   &   0   &   GI_{3}
\end{pmatrix}
$$
and  $G$ and $E$ are the shear and bulk modulus, whereas $A$ and $I_i$ are the area and quadratic moment (along $\bvec{d_i}$) of the cross-section. Note that $\mathbb{G}$ and $\mathbb{H}$ are both constant and uniform. This excludes the case of a non-homogeneous beam to the present work.\\
The internal force $\bvec{N}$ and the torque $\bvec{M}$ acting on a beam section are defined in the following way:
\begin{equation}\label{NMPsirelationGenereal}
\bvec{N}:=\frac{\partial U }{\partial \boldsymbol{\varepsilon}}, 
\quad\quad\quad
\bvec{M}:=\frac{\partial U }{\partial \boldsymbol{\kappa}}.
\end{equation}
From \eqref{PsiGeneral} and \eqref{NMPsirelationGenereal}, the linear stress-strain relations are:
\begin{equation}\label{LinearStressStrainRelation}
\bvec{N}=\mathbb{G}\boldsymbol{\varepsilon}, 
\quad\quad\quad
\bvec{M}=\mathbb{H}\boldsymbol{\kappa}.
\end{equation}

\subsection{Kinetics}
The kinetics of the beam is defined by the spin $\boldsymbol{\omega}$ of the cross-section and the velocity $\bvec{v}$ of the center of mass: 
\begin{equation}\label{Vitesse}
        \bvec{v}:=\frac{\partial \boldsymbol{\varphi}}{\partial t}.
        \end{equation}
One introduces the kinetic energy density $T(\bvec{v}, \boldsymbol{\omega})$
\begin{equation}\label{eq:kinetic_energy}    
T:=
\frac{1}{2}\bvec{v} \mathbb{A}\bvec{v} 
+
\frac{1}{2} \boldsymbol{\omega}\mathbb{J}\boldsymbol{\omega}.
\end{equation}
where $\mathbb{A}$ and $\mathbb{J}$
are diagonal inertial tensors in the mobile frame $\{\bvec{d_{i}}\}$:
$$
\mathbb{A}:=
\begin{pmatrix}
    \rho A  &   0   &   0 \\
    0   &   \rho A  &   0 \\
    0   &   0   &   \rho A
\end{pmatrix},
\quad\quad\quad
\mathbb{J}:=
\begin{pmatrix}
    \rho I_{1}  &   0   &   0 \\
    0   &   \rho I_{2}  &   0 \\
    0   &   0   &   \rho I_{3}
\end{pmatrix},
$$
$\rho$ being the mass density. As for $\mathbb{G}$ and $\mathbb{H}$, the inertial tensors $\mathbb{A}$ and $\mathbb{J}$ are both constant and uniform. One considers indeed a uniform beam.
\subsection{Strong formulation of dynamical equations}

\subsubsection{In the mobile frame}

The general equilibrium relations of Timoshenko model are (\textit{e.g.} \cite{le2018vibration,simo1986three}):
\begin{align}
\label{PFD1}
\begin{split}
\dsp\frac{\partial\bvec{N}}{\partial S}	&=\frac{\partial \, \mathbb{A} \bvec{v} }{\partial t} \\[5pt]
\dsp \frac{\partial \bvec{M}}{\partial S}+  \frac{\partial \boldsymbol{\varphi}}{\partial S} \wedge\bvec{N}	&=\frac{\partial \, \mathbb{J}\boldsymbol{\omega} }{\partial t}
\end{split}
\end{align}
However, in the present context, this system cannot be solved without introducing additional relation. The strain definition \eqref{StrainDefinition} and the stress-strain relations \eqref{LinearStressStrainRelation} are involved to obtain:
\begin{align}
\label{eqPhys2}
\begin{split}
\dsp\frac{\partial \, \mathbb{G}\boldsymbol{\varepsilon}}{\partial S}	&=\frac{\partial \, \mathbb{A} \bvec{v} }{\partial t}\\[5pt]
\dsp \frac{\partial \,\mathbb{H}\boldsymbol{\kappa}}{\partial S}+  \left(\boldsymbol{\varepsilon}+\bvec{d_3}\right) \wedge \left(\mathbb{G}\boldsymbol{\varepsilon}\right)	&=\frac{\partial \,  \mathbb{J}\boldsymbol{\omega} }{\partial t}.
\end{split}
\end{align}
\begin{rem}[Static or rigid problem]
This problem can generally be solved for quasi-static case, where the right-hand side is zero, if boundary conditions are supplied. Indeed, this quasi-static problem decomposes in the frame $\left\{\bvec{d_i}\right\}$, in a system of six-scalar ordinary, first-order, differential equations, with six unknowns $(\varepsilon_i)$ and $(\kappa_i)$. This remains true for rigid situation, where the left-hand side -- instead of the right-hand side -- is zero. This last problem is nothing else than the standard Euler dynamic equations of the rigid body. 
\end{rem}

\begin{rem}\label{remScharzVector}
In order to complete the formulation \eqref{eqPhys2}, the relation \eqref{crossderivation2} gives a first additional relation. A second relation is:
\begin{equation}\label{closurerelations2}
\frac{\partial \bvec{v}}{\partial S} -\boldsymbol{\omega}\wedge \bvec{d_3}=\frac{\partial \boldsymbol{\varepsilon} }{\partial t}, 
\end{equation}
coming from crossed derivation of $\boldsymbol{\varphi}$:
\begin{equation}
\frac{\partial }{\partial S}\frac{\partial \boldsymbol{\varphi}}{\partial t} = \frac{\partial }{\partial t}\frac{\partial \boldsymbol{\varphi}}{\partial S}    
\end{equation}
and using \eqref{StrainDefinition} and \eqref{Vitesse}:
\begin{align*}\label{crossderivationVarPhi}
\frac{\partial \bvec{v}}{\partial S}&= \frac{\partial }{\partial t}\left(\boldsymbol{\varepsilon}+\bvec{d_3}\right)\\
&=\frac{\partial \boldsymbol{\varepsilon} }{\partial t}+\boldsymbol{\omega}\wedge\bvec{d_3}.
\end{align*}
\end{rem}

From \eqref{crossderivation2}, \eqref{eqPhys2} and \eqref{closurerelations2}, a system of 4 first-order differential equations with 4 $3$-dimensional unknowns is addressed:
\begin{equation}\label{SystemPropre}
\left\{    \begin{array}{rcl}
\dsp
\frac{\partial \, \mathbb{G}\boldsymbol{\varepsilon}}{\partial S}	
&=&\dsp
\frac{\partial  \,  \mathbb{A} \bvec{v}}{\partial t}\\[6pt]
\dsp 
\frac{\partial \,\mathbb{H}\boldsymbol{\kappa}}{\partial S}+  \left(\boldsymbol{\varepsilon}+\bvec{d_3}\right) \wedge \left(\mathbb{G}\boldsymbol{\varepsilon}\right)	
&=&\dsp
\frac{\partial  \,  \mathbb{J}\boldsymbol{\omega}}{\partial t}
\\[6pt]
\dsp
\frac{\partial \, \bvec{v}}{\partial S} -\boldsymbol{\omega}\wedge \bvec{d_3}
&=&\dsp
\frac{\partial \, \boldsymbol{\varepsilon} }{\partial t}  
\\[6pt]
\dsp
\frac{\partial \, \boldsymbol{\omega} }{\partial S} + \boldsymbol{\omega}\wedge \boldsymbol{\kappa}
&=&\dsp
\frac{\partial  \,  \boldsymbol{\kappa} }{\partial t} 
    \end{array}
\right.    
\end{equation}
The two last relations are sometimes called \emph{closure} (or \emph{compatibility}) relations. 
\begin{rem}
    After projection along $\{\bvec{d_{i}}\}$, this system allows to solve $(v_i)$, $(\omega_i)$, $(\varepsilon_i)$ and $(\kappa_i)$, if boundary conditions and initial conditions are given properly. This corresponds indeed to solve equation \eqref{SystemPropre} within an \emph{algebraic} formulation.
\end{rem}
 \subsubsection{With algebraic formulation}
The system \eqref{SystemPropre} gets written in an \emph{algebraic} formalism as
\begin{equation}
    \label{SystemPropreM0}
\left\{    \begin{array}{rcl}
\dsp
\frac{\partial \, \mathbb{G}\varepsilon}{\partial S} + \kappa \wedge \mathbb{G}\varepsilon
-\omega \wedge \mathbb{A}v&=&\dsp
\frac{\partial  \,  \mathbb{A} v}{\partial t} 	\\[6pt]
\dsp 
\frac{\partial \,\mathbb{H}\kappa}{\partial S} +\kappa \wedge \mathbb{H}\kappa+  \left(\varepsilon+d_3 \right) \wedge \mathbb{G}\varepsilon
-\omega\wedge \mathbb{J}\omega
&=&\dsp
\frac{\partial  \,  \mathbb{J}\omega}{\partial t} 
\\[6pt]
\dsp
\frac{\partial \, v}{\partial S} + \kappa\wedge v
+ (\varepsilon+d_{3}) \wedge \omega 
&=&\dsp
\frac{\partial \, \varepsilon}{\partial t}  
\\[6pt]
\dsp
\frac{\partial \, \omega }{\partial S}+\kappa\wedge\omega
&=&\dsp
\frac{\partial  \,  \kappa }{\partial t} 
\end{array}
\right.    
\end{equation}
\begin{rem}
More concisely, and after multiplying the third line by $\mathbb{G}$ and the last one by $\mathbb{H}$:
\begin{equation}
    \label{SystemPropreM1}
\begin{pmatrix}
0			&0			&\mathbb{G}	&0				\\
0			&0			&0			&\mathbb{H}		\\
\mathbb{G}	&0			&0			&0				\\
0			&\mathbb{H}	&0			&0				
\end{pmatrix}
\begin{pmatrix}
v'	\\
\omega'	\\
\varepsilon'	\\
\kappa'		\\
\end{pmatrix}
+
\begin{pmatrix}
-\mathbb{W}\mathbb{A}	&0					&\mathbb{K}\mathbb{G}	&0					\\
0					& \mathbb{F}	&\mathbb{E}\mathbb{G}	&\mathbb{K}\mathbb{H}	\\
\mathbb{G}\mathbb{K}			&\mathbb{G}\mathbb{E}			&0			&0					\\
0					&\mathbb{H}\mathbb{K}			&0					&0						\\
\end{pmatrix}
\begin{pmatrix}
v	\\
\omega	\\
\varepsilon	\\
\kappa		
\end{pmatrix}
=
\begin{pmatrix}
\mathbb{A}	&0			&0			&0\\
0			&\mathbb{J}	&0			&0	\\
0			&0			&\mathbb{G}	&0			\\
0			&0			&0			&\mathbb{H}
\end{pmatrix}
\begin{pmatrix}
\dot{v}	\\
\dot{\omega}	\\
\dot{\varepsilon}	\\
\dot{\kappa}		
\end{pmatrix}
\end{equation}
where $\mathbb{F}:= j^{-1}(\mathbb{J}\boldsymbol{\omega})$ and $\mathbb{E}:= j^{-1}(\boldsymbol{\varepsilon} + \bvec{d_3})$, then 
$$
\mathbb{F}=
\begin{pmatrix}
0                   &	-\rho I_3\omega_{3}   & \rho I_2\omega_{2}	\\
\rho I_3\omega_{3}  &	0     				  &-\rho I_1\omega_{1}	\\
-\rho I_2\omega_{2}	&	\rho I_1\omega_{1}	  & 0
\end{pmatrix},
\quad \quad \quad 
\mathbb{E}
\begin{pmatrix}
0				&	-(\varepsilon_{3}+1)	&\varepsilon_{2}	\\
\varepsilon_{3}+1	&	0				&-\varepsilon_{1}	\\
-\varepsilon_{2}		&	\varepsilon_{1}		&0
\end{pmatrix}
$$
respectively. As the other matrices, they have their components given in the moving frame $\{\bvec{d_i}\}$. Vector's component derivations are written with the following convention (here for $v$)
$$
v'=
\begin{pmatrix}
\frac{\partial v_{1}}{\partial S}	\\[4pt]
\frac{\partial v_{2}}{\partial S}	\\[4pt]
\frac{\partial v_{3}}{\partial S}	\\
\end{pmatrix}
\quad\quad\quad
\dot{v}=
\begin{pmatrix}
\frac{\partial v_{1}}{\partial t}	\\[4pt]
\frac{\partial v_{2}}{\partial t}	\\[4pt]
\frac{\partial v_{3}}{\partial t}	\\
\end{pmatrix}.
$$
Let us introduce 
$$
u:=\begin{pmatrix}
    v\\
    \omega\\
    \varepsilon\\
    \kappa
\end{pmatrix},
\quad
\mathcal{D}:=
\begin{pmatrix}
0			&0			&\mathbb{G}	&0				\\
0			&0			&0			&\mathbb{H}		\\
\mathbb{G}	&0			&0			&0				\\
0			&\mathbb{H}	&0			&0
\end{pmatrix},
\quad
\mathcal{Y}:=
\begin{pmatrix}
-\mathbb{W}\mathbb{A}	&0					&\mathbb{K}\mathbb{G}	&0					\\
0					&\mathbb{F}	&\mathbb{E}\mathbb{G}	&\mathbb{K}\mathbb{H}	\\
\mathbb{G}\mathbb{K}			&\mathbb{G}\mathbb{E}			&0			&0					\\
0					&\mathbb{H}\mathbb{K}			&0					&0						\\
\end{pmatrix},
\quad
\mathcal{M}:=
\begin{pmatrix}
\mathbb{A}	&0			&0			&0\\
0			&\mathbb{J}	&0			&0	\\
0			&0			&\mathbb{G}	&0			\\
0			&0			&0			&\mathbb{H}
\end{pmatrix}
$$
where $\mathcal{D}$ is symmetric and $\mathcal{M}$ is diagonal; both are constant and uniform, this is not the case of the skew-symmetric matrix $\mathcal{Y}$ that is $u(S,t)$-dependent. The system \eqref{SystemPropreM1} becomes a \emph{transport equation}:
\begin{align}\label{SystemPropreM2}
\mathcal{D}u'+\mathcal{Y}u=\mathcal{M}\dot{u}
\end{align}
\end{rem}

\subsection{Determining the kinematics}\label{sec:kinematics}
We highlight that solving \eqref{SystemPropre} (or \eqref{SystemPropreM2}) is not sufficient to determine $\{\bvec{v},\, \boldsymbol{\omega},\, \boldsymbol{\varepsilon},\, \boldsymbol{\kappa} \}.$ Indeed, even if components $\{v, \, \omega, \, \varepsilon, \, \kappa \}$ are known, it is not the case of the local frame $\{\bvec{d_{i}}\}$. \\
Therefore one has to determine not only $\bvec{R}(S,t)$ or $\{\bvec{d_{i}}(S,t)\}$ (being equivalent according to \eqref{DEFrotation}) but also the placement $\boldsymbol{\varphi}(S,t)$ (or $\varphi(S,t)$) too. Let us first begin with this latter. 
\subsubsection{Determining the placement}
The placement $ \boldsymbol{\varphi}$ has to satisfy both \eqref{StrainDefinition} and \eqref{Vitesse}. The \emph{closure} relation \eqref{closurerelations2} ensures that if one of this relation is satisfied, the other is satisfied up to a constant.  Then $ \boldsymbol{\varphi}(S,t)$ has to satisfy  
$$
\textrm{\underline{either}}\quad\quad \frac{\partial\, \boldsymbol{\varphi} }{\partial S}=\boldsymbol{\varepsilon}+\bvec{d_{3}}
\quad\quad\quad\textrm{\underline{or}}\quad\quad
\frac{\partial\, \boldsymbol{\varphi} }{\partial t}=\bvec{v}.
$$
As $\{v, \, \omega, \, \varepsilon, \, \kappa \}$ are supposed to be known, these relations write:
$$
\varphi'+\kappa\wedge\varphi=\varepsilon+d_{3}
\quad\quad\quad\textrm{or}\quad\quad
\dot{\varphi}+\omega\wedge\varphi=v.
$$
Any of these differential relations can be used to determine $\varphi(S,t)$. More precisely, left-side case can be solved if a boundary condition $\varphi(0,t)$ is given and right-side case can be solved if an initial condition $\varphi(S,0)$ is prescribed. 
\subsubsection{Determining the orientation of the cross-section}
Determining the frame $\{\bvec{d_{i}}\}$ can be explicitly performed if an other frame is known. In our case the Cartesian frame $\{\bvec{e_{I}}\}$ is chosen as reference. In other words, one has to determine the components $(d_{Ii}):=(d_{xi},d_{yi},d_{zi})$ of $\bvec{d_{i}}=d_{Ii}\bvec{e_{I}}$. In the Cartesian frame, solving 
$$
\frac{\partial \bvec{d_{i}}}{\partial t}=\boldsymbol{\omega}\wedge\bvec{d_{i}}
$$
is equivalent to solve 
$$
\frac{\partial}{\partial t}
\begin{pmatrix}
d_{x1}	&	d_{x2}	&	d_{x3}	\\
d_{y1}	&	d_{y2}	&	d_{y3}	\\
d_{z1}	&	d_{z2}	&	d_{z3}	
\end{pmatrix}
=
\begin{pmatrix}
d_{x1}	&	d_{x2}	&	d_{x3}	\\
d_{y1}	&	d_{y2}	&	d_{y3}	\\
d_{z1}	&	d_{z2}	&	d_{z3}	
\end{pmatrix}
\begin{pmatrix}
0	&	\omega_{3}	&	-\omega_{2}	\\
-\omega_{3}	&	0	&	\omega_{1}	\\
\omega_{2}	&	-\omega_{1}	&	0	
\end{pmatrix}
$$
where the matrix $\mathbb{R}:= [d_{Ii}]$ is nothing else than the matrix representation of $\bvec{R}$ in the Cartesian frame $\{\bvec{e_I}\}$. This provides the concise relation:
$$
\frac{\partial \mathbb{R}}{\partial t} - \mathbb{R} \mathbb{W}=0
$$
Solving this systems allows to determine $\bvec{R}$ and the components of $\{\bvec{d_{i}}\}$ in the Cartesian frame, if initial conditions are given for the directors along $S$. \\
As for $\boldsymbol{\varphi}$, another integration choice can be used by exploiting space derivation. In that case $\bvec{R}$ and $\{\bvec{d_{i}}\}$ can be determined by solving 
$$
\frac{\partial \mathbb{R}}{\partial S} - \mathbb{R}\mathbb{K}=0
$$
if a boundary condition is given for the frame $\{\bvec{d_{i}}(0,t)\}$ at any time. 
\subsection{Conservation of energy}\label{sec:conservation}
In this section, we use the formulation \eqref{SystemPropreM2} to give a concise proof of the conservation of energy in Timoshenko model.
\begin{theorem}
    The conservation of energy (strain energy and kinetic energy) is ensured if boundary conditions are chosen properly. Using components, this is formulated by
    \begin{align}\label{eqPhys2M}
\frac{\partial}{\partial t}
\int_{0}^{L}
\frac{1}{2}{v}\mathbb{A}{v} 
+
\frac{1}{2}{\omega}\mathbb{J}{\omega} 
+
\frac{1}{2}{\varepsilon}\mathbb{G}{\varepsilon}
+
\frac{1}{2}{\kappa}\mathbb{H}{\kappa}
\, \dd S
=
\Big[
v\mathbb{G}\varepsilon
+
\omega\mathbb{H}\kappa
\Big]_{0}^{L}.
\end{align}
\end{theorem}
\begin{proof}
By projecting equation \eqref{SystemPropreM2} along the vector $u$ through the usual Euclidean scalar product and integrating along the fiber:
\begin{equation}\label{ProjectionU}
    \int_0^L u \cdot \mathcal{D}u'+ u \cdot \mathcal{Y}u - u \cdot \mathcal{M}\dot{u} \,\dd S = 0.
\end{equation}
First, direct integration delivers:
\begin{align*}
    \int_0^L u \cdot \mathcal{D} u^{'} \text{d}S &= \int_0^L v \cdot \mathbb{G} \varepsilon' + \omega \cdot \mathbb{H} \kappa' + \varepsilon \cdot \mathbb{G} v' + \kappa \cdot \mathbb{H} \omega' \text{d}S\\
    &= \left[v \mathbb{G} \varepsilon + \omega \mathbb{H} \kappa\right]_0^L.
\end{align*}
The second term vanishes thanks to the property of vectorial product:
\begin{align*}
    \int_0^L u \cdot \mathcal{Y} u \text{d}S &= -\int_0^L v \cdot \mathbb{W}\mathbb{A} v + \omega \cdot \mathbb{F} \omega \text{d}S\\
    &= -\int_0^L (\mathbb{A}v) \cdot ( \omega \wedge  v ) + \omega \cdot (\omega \wedge (\mathbb{J} \omega ) ) \text{d}S\\
    &=0.
\end{align*}
The last part provides:
\begin{equation*}
    \int_0^L u \cdot \mathcal{M} \dot{u} \dd S = \frac{\partial }{\partial t}  \int_0^L \frac{1}{2}\left( v \cdot \mathbb{A} v + \omega \cdot \mathbb{J} \omega + \varepsilon \cdot  \mathbb{G} \varepsilon + \kappa  \cdot \mathbb{H} \kappa \right)\text{d} S.
\end{equation*}
Hence \eqref{ProjectionU} gives the desired result.
\end{proof}

\begin{rem}[Mobile reformulation] The relation \eqref{eqPhys2M} has the mobile frame counterpart
\begin{align}\label{eqPhys2V}
\frac{\partial}{\partial t}
\int_{0}^{L}
\frac{1}{2}\bvec{v}\mathbb{A}\bvec{v} 
+
\frac{1}{2}\boldsymbol{\omega}\mathbb{J}\boldsymbol{\omega} 
+
\frac{1}{2}\boldsymbol{\varepsilon}\mathbb{G}\boldsymbol{\varepsilon}
+
\frac{1}{2}\boldsymbol{\kappa}\mathbb{H}\boldsymbol{\kappa}
\, \dd S
=
\Big[
\bvec{v}\mathbb{G}\boldsymbol{\varepsilon}
+
\boldsymbol{\omega}\mathbb{H}\boldsymbol{\kappa}
\Big]_{0}^{L}
\end{align}
\end{rem}
Moreover, particular attention must be paid on time derivation of strain energy. Indeed, quadratic terms in mobile and algebraic formulation are derived the same way with respect to time:
\begin{equation}
    \begin{array}{llll}
     \dsp \frac{\partial }{\partial t}\left(\frac{1}{2}\boldsymbol{\varepsilon}\mathbb{G}\boldsymbol{\varepsilon}\right)&=\dsp 
\left( \frac{\partial \, \boldsymbol{\varepsilon} }{\partial t}    -   
\boldsymbol{\omega}  \wedge \boldsymbol{\varepsilon}\right) 	  \cdot (\mathbb{G}\boldsymbol{\varepsilon})
&=\dsp 
\frac{\partial \, \varepsilon }{\partial t}   \cdot (\mathbb{G}\varepsilon)
&=\dsp 
\frac{\partial}{\partial t} \left( \frac{1}{2}\varepsilon \mathbb{G}\varepsilon\right)\\
    \dsp \frac{\partial }{\partial t}\left(\frac{1}{2}\boldsymbol{\kappa}\mathbb{H}\boldsymbol{\kappa}\right)&=\dsp 
\left( \frac{\partial \, \boldsymbol{\kappa} }{\partial t}    -   
\boldsymbol{\omega}  \wedge \boldsymbol{\kappa}\right) 	  \cdot (\mathbb{H}\boldsymbol{\kappa})
&=\dsp 
\frac{\partial \, \kappa }{\partial t}   \cdot (\mathbb{H}\kappa)
&=\dsp 
\frac{\partial}{\partial t} \left( \frac{1}{2}\kappa \mathbb{H}\kappa\right) 
\end{array}
\end{equation}
due to symmetry of rigidity tensors $\mathbb{G}$ and $\mathbb{H}$.

\subsection{Variational principle}\label{sec:VariationalPrinciple1}

The Lagrangian density $\ell(\bvec{v},\boldsymbol{\omega},\boldsymbol{\varepsilon},\boldsymbol{\kappa})$ and Lagrangian  $\mathcal{L}(\bvec{v},\boldsymbol{\omega},\boldsymbol{\varepsilon},\boldsymbol{\kappa})$ 
associated to this problem are  in the \emph{mobile} convention:
\begin{align}\label{LagrangianV}
\ell(\bvec{v},\boldsymbol{\omega},\boldsymbol{\varepsilon},\boldsymbol{\kappa}) &:= 
\frac{1}{2}\bvec{v}\mathbb{A}\bvec{v} 
+
\frac{1}{2}\boldsymbol{\omega}\mathbb{J}\boldsymbol{\omega} 
-
\frac{1}{2}\boldsymbol{\varepsilon}\mathbb{G}\boldsymbol{\varepsilon}
-
\frac{1}{2}\boldsymbol{\kappa}\mathbb{H}\boldsymbol{\kappa},
&
\mathcal{L}(\bvec{v},\boldsymbol{\omega},\boldsymbol{\varepsilon},\boldsymbol{\kappa}) &:= 
\int_{0}^{L}
\ell(\bvec{v},\boldsymbol{\omega},\boldsymbol{\varepsilon},\boldsymbol{\kappa}) 
\, \dd S
\end{align}
Accordingly, the \emph{action} $\mathcal{S}$ is:
$$
\mathcal{S}:=\int_{t_{1}}^{t_{2}}\mathcal{L}\, \dd t.
$$
\begin{theorem}\label{thm:weak_formulation}
Under suitable boundary conditions, Hamilton's Principle $\delta\mathcal{S}=0$ is equivalent to 
\begin{align}\label{eq:importante}
\int_{t_{1}}^{t_{2}} \int_{0}^{L}
\delta\boldsymbol{\varphi}\cdot \left(
\frac{\partial \, \mathbb{G}\boldsymbol{\varepsilon}}{\partial S}	
-\frac{\partial \, \mathbb{A}\bvec{v}}{\partial t}
\right)
+
 \delta\boldsymbol{\theta}\cdot\left(
 \frac{\partial\, \mathbb{H}\boldsymbol{\kappa}}{\partial S}
 +
\frac{\partial \boldsymbol{\varphi}}{\partial S}\wedge(\mathbb{G}\boldsymbol{\varepsilon})
-
\frac{\partial \, \mathbb{J}\boldsymbol{\omega}}{\partial t } 
\right)
\, \dd S \dd t	= 0.
\end{align}
\end{theorem}

\begin{proof}
Perturbation of the action
\begin{align}\label{LeastActionPrinciple}
\int_{t_{1}}^{t_{2}} \int_{0}^{L}\delta\ell\, \dd S \dd t &= 0.
\end{align}
is performed according to \eqref{PerturbHermitianProd} applied on each quadratic forms, hence \eqref{LeastActionPrinciple} becomes
\begin{align}\label{eq:ham_princ_var1}
\begin{split}
\int_{t_{1}}^{t_{2}} \int_{0}^{L}
\left(\mathbb{A}\bvec{v}\right)\cdot\left(\delta \bvec{v}-\delta\boldsymbol{\theta}\wedge\bvec{v} \right) 
-\left(\mathbb{G}\boldsymbol{\varepsilon}\right)\cdot\left(\delta\boldsymbol{\varepsilon}-\delta\boldsymbol{\theta}\wedge\boldsymbol{\varepsilon}\right)	\\
+
\left(\mathbb{J}\boldsymbol{\omega}\right)\cdot\Big(\delta\boldsymbol{\omega}-\delta\boldsymbol{\theta}\wedge\boldsymbol{\omega}\Big)-
\left(\mathbb{H}\boldsymbol{\kappa}\right)\cdot\Big(\delta\boldsymbol{\kappa}-\delta\boldsymbol{\theta}\wedge\boldsymbol{\kappa}\Big)
\, \dd S \dd t   &= 0.
\end{split}
\end{align}
From Schwarz Lemma,
\begin{align}\label{unpetitdebut}
\delta\bvec{v}
=\frac{\partial \delta\boldsymbol{\varphi}}{\partial t},
\quad\quad
\delta\boldsymbol{\varepsilon}&=\frac{\partial \delta\boldsymbol{\varphi}}{\partial S}-\delta\bvec{d_3}
=\frac{\partial \delta\boldsymbol{\varphi}}{\partial S}-\delta\boldsymbol{\theta}\wedge\bvec{d_3}
\end{align}
and one obtains thanks to \eqref{eq:var_omega_kappa}:
\begin{align}
\begin{split}
\int_{t_{1}}^{t_{2}} \int_{0}^{L}
\left(\mathbb{A}\bvec{v}\right)\cdot\left(\frac{\partial \delta\boldsymbol{\varphi}}{\partial t}-\delta\boldsymbol{\theta}\wedge\frac{\partial \boldsymbol{\varphi}}{\partial t}\right) 
-\left(\mathbb{G}\boldsymbol{\varepsilon}\right)\cdot\left(\frac{\partial \delta\boldsymbol{\varphi}}{\partial S}-\delta\boldsymbol{\theta}\wedge\frac{\partial \boldsymbol{\varphi}}{\partial S}\right)	\\
+
\left(\mathbb{J}\boldsymbol{\omega}\right)\cdot\frac{\partial\, \delta\boldsymbol{\theta}}{\partial t}-
\left(\mathbb{H}\boldsymbol{\kappa}\right)\cdot\frac{\partial\, \delta\boldsymbol{\theta}}{\partial S}
\, \dd S \dd t   &= 0.	\\
\end{split}
\end{align}
Since $\mathbb{A}=\rho A  \mathbb{1}$, $\mathbb{A}\bvec{v}$ and $\frac{\partial \boldsymbol{\varphi}}{\partial t}=\bvec{v}$ are colinear. Therefore,
\begin{align}
\left(\mathbb{A}\bvec{v}\right)\cdot\left(\delta\boldsymbol{\theta}\wedge\frac{\partial \boldsymbol{\varphi}}{\partial t}\right) =0.
\end{align}
As a consequence:
\begin{align}\label{eq:ham_princ_2}
\begin{split}
\int_{t_{1}}^{t_{2}} \int_{0}^{L}
\frac{\partial \, \delta\boldsymbol{\varphi}}{\partial t}\cdot(\mathbb{A}\bvec{v}) - 
\frac{\partial \, \delta\boldsymbol{\varphi}}{\partial S}\cdot(\mathbb{G}\boldsymbol{\varepsilon})	
+
\left(\delta\boldsymbol{\theta}\wedge\frac{\partial \, \boldsymbol{\varphi}}{\partial S}\right)\cdot(\mathbb{G}\boldsymbol{\varepsilon})
+
\frac{\partial \, \delta\boldsymbol{\theta}}{\partial t }\cdot(\mathbb{J}\boldsymbol{\omega}) -
\frac{\partial \, \delta\boldsymbol{\theta}}{\partial S} \cdot(\mathbb{H}\boldsymbol{\kappa})
\, \dd S \dd t   &= 0	
\end{split}
\end{align}
By integration by parts:
\begin{align}
\begin{split}
\int_{t_{1}}^{t_{2}} \int_{0}^{L}
- \delta\boldsymbol{\varphi}\cdot \frac{\partial \, \mathbb{A}\bvec{v}}{\partial t}
+ 
\delta\boldsymbol{\varphi}\cdot\frac{\partial \, \mathbb{G}\boldsymbol{\varepsilon}}{\partial S}	
+
\left(\delta\boldsymbol{\theta}\wedge\frac{\partial \boldsymbol{\varphi}}{\partial S}\right)\cdot(\mathbb{G}\boldsymbol{\varepsilon})
-
\delta\boldsymbol{\theta}\cdot\frac{\partial \, \mathbb{J}\boldsymbol{\omega}}{\partial t } 
+
\ \delta\boldsymbol{\theta} \cdot \frac{\partial\, \mathbb{H}\boldsymbol{\kappa}}{\partial S}
\, \dd S \dd t   &	\\
+\int_{0}^{L}
\Big[
\delta\boldsymbol{\varphi}\cdot(\mathbb{A}\bvec{v})
+
\delta\boldsymbol{\theta}\cdot(\mathbb{J}\boldsymbol{\omega})
 \Big]_{t_{1}}^{t_{2}}\,\dd S
-
\int_{t_{1}}^{t_{2}}
\Big[
 \delta\boldsymbol{\varphi} \cdot (\mathbb{G}\boldsymbol{\varepsilon})
 +
\delta\boldsymbol{\theta} \cdot (\mathbb{H}\boldsymbol{\kappa})
 \Big]_{0}^{L}\,\dd t&=0.
\end{split}
\end{align}
Since $\delta\boldsymbol{\varphi}(S,t_{i})=0$ and $\delta\boldsymbol{\theta}(S,t_{i})=0$ for $t_{i} \in \{t_{1},\, t_{2}\},$ one obtains after a cyclic permutation of the mixed product:
\begin{align}
\begin{split}
\int_{t_{1}}^{t_{2}} \int_{0}^{L}
\delta\boldsymbol{\varphi}\cdot \left(
\frac{\partial \, \mathbb{G}\boldsymbol{\varepsilon}}{\partial S}	
-\frac{\partial \, \mathbb{A}\bvec{v}}{\partial t}
\right)
+
 \delta\boldsymbol{\theta}\cdot\left(
 \frac{\partial\, \mathbb{H}\boldsymbol{\kappa}}{\partial S}
 +
\frac{\partial \boldsymbol{\varphi}}{\partial S}\wedge(\mathbb{G}\boldsymbol{\varepsilon})
-
\frac{\partial \, \mathbb{J}\boldsymbol{\omega}}{\partial t } 
\right)
\, \dd S \dd t   &	=\quad \quad \hfill \\
\quad \hfill \quad \quad \quad \int_{t_{1}}^{t_{2}}
\Big[
 \delta\boldsymbol{\varphi} \cdot (\mathbb{G}\boldsymbol{\varepsilon})
 +
\delta\boldsymbol{\theta} \cdot (\mathbb{H}\boldsymbol{\kappa})
 \Big]_{0}^{L}\,\dd t
\end{split}
\end{align}
\end{proof}

\begin{corollary}[Strong formulation]
Under suitable boundary conditions, Hamilton's Principle $\delta\mathcal{S}=0$ is equivalent to the strong formulation
\begin{align}\label{eq:equillibrium}
\begin{split}
\frac{\partial \, \mathbb{G}\boldsymbol{\varepsilon}}{\partial S}	
&=
\frac{\partial \, \mathbb{A}\bvec{v}}{\partial t}
\\
 \frac{\partial\, \mathbb{H}\boldsymbol{\kappa}}{\partial S}
 +
\frac{\partial \boldsymbol{\varphi}}{\partial S}\wedge(\mathbb{G}\boldsymbol{\varepsilon})
&=
\frac{\partial \, \mathbb{J}\boldsymbol{\omega}}{\partial t }.
\end{split}
\end{align}
\end{corollary}

\begin{proof}
We choose test functions $\delta \boldsymbol{\theta}$ and $\delta \boldsymbol{\varphi}$ such that, at any time $t$ and at $S=0$ and $S=L,$
\begin{equation}\label{eq:BCCool}
    <\delta \boldsymbol{\varphi}, \mathbb{G}\boldsymbol{\varepsilon}> = <\delta\boldsymbol{\theta}, \mathbb{H}\boldsymbol{\kappa}> = 0.
\end{equation}
As this relation is true at any time $t_1$ and $t_2$ and for any $\delta \boldsymbol{\theta}$ and $\delta \boldsymbol{\varphi},$ one gets out of theorem \ref{thm:weak_formulation} the desired strong formulation \eqref{eq:equillibrium}.
\end{proof}
In that sense the variable $\delta \boldsymbol{\theta}$ -- involved in the weak formulation \eqref{eq:importante} -- is related to Euler-Poincar\'e reduction \cite{poincare1901forme}. 

\section{Hamiltonian mechanics of Timoshenko model}\label{sec:ham_mechanics_tim}

The Lagrangian has been initially introduced with secondary variables: the velocity $\bvec{v},$ the spin $\boldsymbol{\omega},$ the curvature $\boldsymbol{\kappa}$ and the strain $\boldsymbol{\varepsilon}.$ However, weak formulation \eqref{eq:importante} and equilibrium relations \eqref{eq:equillibrium} have been obtained by introducing test functions $\delta \boldsymbol{\theta}$ and $\delta \boldsymbol{\varphi}$ associated to primary kinematical variables $\boldsymbol{\varphi}$ and $\bvec{R}$. This motivates the study we lead in the present section, considering dual variables in order to obtain three different Hamiltonian formulations of Timoshenko model. This section contains the main results of this work.

\subsection{A first Hamiltonian formulation of Timoshenko model}

It turns out that, by looking for a Hamiltonian formulation of equation \eqref{SystemPropreM0}, \textit{i.e.} by using algebraic notation, the associated Poisson structure has been already presented (see Theorem 6.2 and Corollary 6.3 of \cite{simo1988hamiltonian}). \\
Let us set $p  = \mathbb{A} v$ and $\sigma = \mathbb{J}\omega.$ The system with variables $(p , \sigma, \varepsilon, \kappa)$ gets written as
\begin{equation}\label{eq:comme_mar}
\left\{    \begin{array}{rcl}
\dsp
\frac{\partial \, \mathbb{G}\varepsilon}{\partial S} + \kappa \wedge (\mathbb{G}\varepsilon)
-(\mathbb{J}^{-1} \sigma) \wedge p &=&\dsp
\frac{\partial  \,  p }{\partial t} 	\\[6pt]
\dsp 
\frac{\partial \,\mathbb{H}\kappa}{\partial S} +\kappa \wedge (\mathbb{H}\kappa)+  \left(\varepsilon+d_3 \right) \wedge (\mathbb{G}\varepsilon)
-(\mathbb{J}^{-1}\sigma)\wedge \sigma
&=&\dsp
\frac{\partial  \,  \sigma}{\partial t} 
\\[6pt]
\dsp
\frac{\partial \, \mathbb{A}^{-1} p }{\partial S} + \kappa\wedge (\mathbb{A}^{-1} p )
+ (\varepsilon+d_{3}) \wedge (\mathbb{J}^{-1}\sigma) 
&=&\dsp
\frac{\partial \, \varepsilon}{\partial t}  
\\[6pt]
\dsp
\frac{\partial \, \mathbb{J}^{-1}\sigma }{\partial S}+\kappa\wedge(\mathbb{J}^{-1}\sigma)
&=&\dsp
\frac{\partial  \,  \kappa }{\partial t} 
    \end{array}
\right.     
\end{equation}

\begin{theorem}\label{thm:old_ham_form}
The system \eqref{eq:comme_mar} is Hamiltonian for the Hamiltonian
\begin{equation}\label{eq:Hamilton1}
        H(p , \sigma, \varepsilon, \kappa) = \int_{0}^{L}
\frac{1}{2}{p }\mathbb{A}^{-1}{p } 
+
\frac{1}{2}{\sigma}\mathbb{J}^{-1}{\sigma} 
+
\frac{1}{2}{\varepsilon}\mathbb{G}{\varepsilon}
+
\frac{1}{2}{\kappa}\mathbb{H}{\kappa}
\, \dd S
\end{equation}
and the Poisson bracket, given for any $f$ and $g$ functions of the variables $\{p,\sigma,\varepsilon,\kappa\}$: 
  \begin{subequations}\label{eqn:main}
    \begin{alignat}{2}
  \{f,g\} = \int_0^L &  < \frac{\partial f}{\partial p }, \frac{\partial }{\partial S}\left( \frac{\partial g}{\partial \varepsilon} \right)> - < \frac{\partial g}{\partial p }, \frac{\partial }{\partial S}\left( \frac{\partial f}{\partial \varepsilon} \right)> \label{subeqn:a}\\
    +& < \frac{\partial f}{\partial \sigma}, \frac{\partial }{\partial S}\left( \frac{\partial g}{\partial \kappa} \right)> - < \frac{\partial g}{\partial \sigma}, \frac{\partial }{\partial S}\left( \frac{\partial f}{\partial \kappa} \right)> \label{subeqn:b}\\
    +&<\kappa, \frac{\partial g}{\partial \varepsilon} \wedge \frac{\partial f}{\partial p} - \frac{\partial f}{\partial \varepsilon} \wedge \frac{\partial g}{\partial p}> \label{subeqn:c}\\ 
    +& < \sigma, \frac{\partial g}{\partial \sigma} \wedge \frac{\partial f}{\partial \sigma}>
    +<p , \frac{\partial g}{\partial \sigma} \wedge \frac{\partial f}{\partial p } - \frac{\partial f}{\partial \sigma} \wedge \frac{\partial g}{\partial p }> \label{subeqn:d}\\ 
    +&<\varepsilon + d_3, \frac{\partial g}{\partial \varepsilon} \wedge \frac{\partial f}{\partial \sigma} - \frac{\partial f}{\partial \varepsilon} \wedge \frac{\partial g}{\partial \sigma}>
    +<\kappa, \frac{\partial g}{\partial \kappa} \wedge \frac{\partial f}{\partial \sigma} - \frac{\partial f}{\partial \kappa} \wedge \frac{\partial g}{\partial \sigma}> \dd S. \label{subeqn:e}
    \end{alignat}
  \end{subequations}
\end{theorem}

\begin{proof}[Sketch of the proof]
We rewrite the Poisson bracket using integration by parts and properties of the cross product:
\begin{align}\begin{split}
    \{f,g\} = \int_0^L &  < \frac{\partial f}{\partial p }, \frac{\partial }{\partial S}\left( \frac{\partial g}{\partial \varepsilon} \right) + p \wedge \frac{\partial g}{\partial \sigma} +   \kappa \wedge \frac{\partial g}{\partial \varepsilon}>\\ 
    +& < \frac{\partial f}{\partial \sigma}, \frac{\partial }{\partial S}\left( \frac{\partial g}{\partial \kappa} \right) + \sigma \wedge \frac{\partial g}{\partial \sigma} + (\varepsilon + d_3) \wedge \frac{\partial g}{\partial \varepsilon} + \kappa \wedge \frac{\partial g}{\partial \kappa}>\\
    +& <\frac{\partial f}{\partial \varepsilon}, \frac{\partial }{\partial S}\left(\frac{\partial g}{\partial p }\right) + (\varepsilon + d_3) \wedge \frac{\partial g}{\partial \sigma} + \kappa \wedge \frac{\partial g}{\partial p} >\\ 
    +& <\frac{\partial f}{\partial \kappa},  \frac{\partial }{\partial S}\left( \frac{\partial g}{\partial \sigma}\right) + \kappa \wedge \frac{\partial g}{\partial \sigma}>\dd S.
\end{split}\end{align}
With the derivatives of the Hamiltonian
\begin{equation}
\begin{array}{ll}
    \dsp \frac{\partial H}{\partial p} =\dsp  \mathbb{A}^{-1}p 
    &\quad\quad\quad\quad\quad\quad
    \dsp \frac{\partial H}{\partial \sigma} =\dsp  \mathbb{J}^{-1}{\sigma} \\[8pt]
    \dsp \frac{\partial H}{\partial \varepsilon} =\dsp  \mathbb{G}{\varepsilon}
    &\quad\quad\quad\quad\quad\quad
    \dsp \frac{\partial H}{\partial \kappa} =\dsp  \mathbb{H}{\kappa},
\end{array}
\end{equation}
we recover easily the equations \eqref{eq:comme_mar}.
\end{proof}

\begin{rem}
    Since any Hamiltonian system is conservative, the conservation of energy obtained in section \ref{sec:conservation} is recovered as a consequence of theorem \ref{thm:old_ham_form}.
\end{rem}
\begin{rem}
The two first lines \eqref{subeqn:a} and \eqref{subeqn:b} are the analogs in our context of a canonical Poisson bracket in mechanics and the third line \eqref{subeqn:c} can be interpreted as an interaction term in between the three variables $(p, \varepsilon, \kappa).$ The two last lines \eqref{subeqn:d}-\eqref{subeqn:e} are the canonical Lie-Poisson bracket on the dual of the semidirect product Lie algebra $(\mathbb{R}^3)^{[0,L]} \ltimes \left( (\mathbb{R}^3)^{[0,L]} \times (\mathbb{R}^3)^{[0,L]} \times (\mathbb{R}^3)^{[0,L]} \right)$ where $\sigma$ is the acting variable and the infinitesimal action is given by the cross product.\\
We refer to the articles \cite{MarwanLie,marsden1984semidirect,holm1998euler} for a study of semidirect products occurring in mechanics and to the lecture \cite{marsden1984reduction} for a nice introduction to the topic.
\end{rem}
However, this Poisson bracket carries many terms and could be heavy to manipulate. On behalf of the authors of the present paper, the main reason is that the algebraic notation is used instead of material coordinates. Furthermore, the set of variables $(p , \sigma, \varepsilon, \kappa)$ is not sufficient to determine the kinematics as it as been mentioned in section \ref{sec:kinematics}. In the sequel, we use material coordinates to provide alternative points of view on Hamiltonian formulations of Timoshenko model.

\subsection{Configuration space and Euler-Lagrange equations}
Now, we use the space of position and velocities $TC = \{(\boldsymbol{\varphi}, \bvec{R},\delta\boldsymbol{\varphi}, \delta\bvec{R} )\} $ and the space of position and momenta $T^* C = \{(\boldsymbol{\varphi}, \bvec{R}, \bvec{p} , \Sigma) \}$ of the configuration space $C = \{(\boldsymbol{\varphi}, \bvec{R}) \}$ to explicit the geometric belonging of each variable.  \\
\\
Let us first introduce the space of configurations and its associated space of velocities.

\begin{definition}[Configuration and velocity spaces]\label{def:configuration_space}
    The configuration space is 
    \begin{equation}
        C= \{(\boldsymbol{\varphi},\bvec{R}) \colon [0,L] \to \mathbb{R}^3 \times SO(3) \}.
    \end{equation}
    The velocity space is the tangent bundle $TC$ of $C$:
    \begin{equation}
        TC = \{(\delta\boldsymbol{\varphi}, \delta\bvec{R}) \colon [0,L] \to \mathbb{R}^3 \times T_{\bvec{R}}SO(3), (\boldsymbol{\varphi}, \bvec{R}) \in C\}.
    \end{equation}
\end{definition}

\begin{rem}\label{rem:matrix_so3}
    Writing any element of the Lie group $SO(3)$ as a matrix allows to represent any tangent vector $\delta \bvec{R} \in T_{\bvec{R}}SO(3)$ at $\bvec{R}$ as an element of the form
    \begin{equation}
        \delta \bvec{R} = \frac{\partial \bvec{R}(\nu)}{\partial \nu}_{|\nu=0}
    \end{equation}
    where $(\bvec{R}(\nu))_{\nu  \in \mathbb{R}}$ is any $\nu$-dependent $SO(3)$ matrix valuing $\bvec{R}(0) = \bvec{R}$ at $0.$ 
\end{rem}

At this stage, we introduce the action
\begin{equation}
    \mathcal{S}\left( (\boldsymbol{\varphi}(t),\bvec{R}(t))_{t_1 \leq t \leq t_2} \right) = \int_{t_1}^{t_2} \mathcal{L}(\boldsymbol{\varphi}, \bvec{R}, \frac{\partial \boldsymbol{\varphi} }{\partial t} , \frac{\partial \bvec{R}}{\partial t} ) \dd t  
\end{equation}
where $\mathcal{L} \colon TC \to \mathbb{R}$ is the Lagrangian density. Evaluated on a given path $(\boldsymbol{\varphi}(t),\bvec{R}(t))_{t_1 \leq t \leq t_2},$ $\mathcal{L}$ is defined as
\begin{equation}
\begin{split}
\mathcal{L}(\boldsymbol{\varphi}, \bvec{R}, \frac{\partial \boldsymbol{\varphi} }{\partial t} , \frac{\partial \bvec{R}}{\partial t} ) &= \frac{1}{2}\int_{0}^{L} <\frac{\partial \boldsymbol{\varphi} }{\partial t} ,\mathbb{A}\frac{\partial \boldsymbol{\varphi}}{\partial t} >
+
 <j(\bvec{R}^{-1} \frac{\partial \bvec{R}}{\partial t}  ),\mathbb{J} j(\bvec{R}^{-1} \frac{\partial \bvec{R}}{\partial t}  )> \\
&\quad\quad -<\frac{\partial \boldsymbol{\varphi} }{\partial S} -\boldsymbol{\bvec{d_3}},\mathbb{G}(\frac{\partial \boldsymbol{\varphi} }{\partial S} -\boldsymbol{\bvec{d_3}})>- <j(\bvec{R}^{-1} \frac{\partial \bvec{R}}{\partial S}  ), \mathbb{H} j(\bvec{R}^{-1} \frac{\partial \bvec{R}}{\partial S}  )> \dd S.
\end{split}
\end{equation}

The results of section \ref{sec:VariationalPrinciple1} are reformulated through the following theorem.

\begin{theorem}
The Euler-Lagrange equations associated to the action $S$ are the equilibrium relations \eqref{eq:equillibrium}.
\end{theorem}

\begin{rem}[Boundary conditions]
    In order to achieve the minimization of the action $S,$ we used boundary conditions \eqref{eq:BCCool}.
\end{rem}

\subsection{Momenta through Legendre transform}

Now, we make use of the Legendre transform to define momenta variables $\bvec{p} $ and $\Sigma.$ They belong to the \emph{phase space} $T^*C = \{(\boldsymbol{\varphi}, \bvec{R}, \bvec{p} , \Sigma) \},$ also called the \emph{position-momenta space}.

\begin{definition}\label{def:metrics}
    On $C,$ we define a Riemannian metric $g$ allowing to identify $TC$ with $T^*C.$ For any $(\boldsymbol{\varphi}, \bvec{R}) \in C$ and any $\Bigl((\delta \boldsymbol{\varphi}, \delta \bvec{R}), (\widetilde{\delta \bvec{R}}, \widetilde{\delta \boldsymbol{\varphi}}) \Bigr) \in T_{(\boldsymbol{\varphi},\bvec{R})}C \times T_{(\boldsymbol{\varphi},\bvec{R})}C,$ the metric $g$ is set to be
    \begin{equation}
        g\Bigl((\delta \boldsymbol{\varphi}, \widetilde{\delta \boldsymbol{\varphi}}), (\delta \bvec{R}, \widetilde{\delta \bvec{R}}) \Bigr) = < \delta \boldsymbol{\varphi}, \widetilde{\delta \boldsymbol{\varphi}} >_{[0,L]} + \ll \delta \bvec{R} , \widetilde{\delta \bvec{R}}\gg_{[0,L]}
    \end{equation}
    with
    \begin{equation}
        < \delta \boldsymbol{\varphi}, \widetilde{\delta \boldsymbol{\varphi}} >_{[0,L]} = \int_0^{L} <\delta \boldsymbol{\varphi}(S), \widetilde{\delta \boldsymbol{\varphi}} (S)> \text{d}S,
    \end{equation}
    $<\cdot ,\cdot >$ denoting the usual scalar product on $\mathbb{R}^3,$ and 
    \begin{equation}
        \ll \delta \bvec{R}, \widetilde{\delta \bvec{R}} \gg_{[0,L]} = \int_0^{L} \ll \delta \bvec{R}(S), \widetilde{\delta \bvec{R}} (S) \gg \text{d}S,
    \end{equation}
    $\ll \bvec{A},\bvec{B} \gg = \frac{1}{2}\text{Tr}(\bvec{A} \cdot \bvec{B}^{T})$ denoting the Frobenius  scalar product of matrices.
\end{definition}

\begin{rem}\label{rem:riesz_lemma}
    The metric $g$ is a useful tool to handle geometric objects in order to achieve concrete computations. Using Riesz's Lemma, definition \ref{def:metrics} provides a representation for covectors in $T^*C.$ Let $\Sigma$ be any covector above $\bvec{R}.$ Thanks to the remark \ref{rem:matrix_so3}, $\Sigma$ can be uniquely represented as a matrix element --denoted again by $\Sigma$ with a slight abuse of notations-- of the form 
     \begin{equation}
       \Sigma(S)  = \frac{\partial \bvec{R}(S,\nu)}{\partial \nu}_{|\nu=0}
    \end{equation}
    where $(\bvec{R}(S,\nu))_{\nu \in \mathbb{R}}$ is any $\nu$-dependent $SO(3)$ matrix valuing $\bvec{R}(S,0) = \bvec{R}(S)$ at $\nu = 0.$ The covector acts on vectors $\delta\bvec{R} $ thanks to the metric:
    \begin{equation}
        \Sigma(\delta\bvec{R}) = \ll \Sigma,\delta\bvec{R}  \gg_{[0,L]}.
    \end{equation}
    A similar remark holds for position momenta $\bvec{p}$: thanks to the metric $g,$ any covector $\bvec{p} $ can be represented as a vector element, denoted $\bvec{p}$ as well. In the sequel, we will make a great use of matrix and vector representations of momenta.
\end{rem}

\begin{rem}\label{rem:p_translation}
    Left translations by elements of the Lie group $SO(3)$ have an interesting interpretation. In the present context, a covector $\Sigma  \in T^*_{\bvec{R}} SO(3)$ is a momentum in Cartesian coordinates, while its left translated element $\bvec{R}^{-1}\Sigma$ lies naturally on the Lie algebra $T^*_{Id} SO(3) \simeq so^* (3).$ $\bvec{R}^{-1}\Sigma$ is then the covector $\Sigma$ translated by $\bvec{R}^{-1}$ on the moving frame. This point is underlined by denoting $\Sigma$ with a normal font (and not bold) as it can be seen as a covector in the Cartesian frame. 
    As already pointed out in \cite{simo1992}, the metric $g$ is invariant with respect to left translations towards the moving frame:
    \begin{equation}
        \ll \delta \bvec{R}, \widetilde{\delta \bvec{R}} \gg_{[0,L]} = \ll \bvec{R}^{-1} \delta \bvec{R}, \bvec{R}^{-1}\widetilde{\delta \bvec{R}} \gg_{[0,L]}.
    \end{equation}
\end{rem}
\begin{rem}    
    Let us also recall here that the left-invariant Frobenius metric \eqref{eq:defmetric} yields to the identification $so(3)^* \simeq so(3)$ allowing to finally identify
\begin{equation}
    T^*_{Id} SO(3) \simeq so(3).
\end{equation}
\end{rem}

To obtain a Hamiltonian formulation, one must study the dynamics on the space of momenta. A fundamental tool to recover equations of motion in terms of the momenta $\bvec{p}  \text{   and   } \Sigma$ is the Legendre transform. 
\begin{prop}\label{LegendreTransform}
    The Legendre transform induced by $\mathcal{L}$ is 
    \begin{equation}
        (q,\delta q ) \in TC \mapsto (q,\pi) \in T^*C
    \end{equation}
    with $q = (\boldsymbol{\varphi}, \bvec{R}) \in C$ and
    \begin{equation}\label{eq:p}
        \pi = \begin{pmatrix}
            \bvec{p}\\ 
            \Sigma
        \end{pmatrix} = \begin{pmatrix}
            \mathbb{A} \delta\boldsymbol{\varphi}\\
            \bvec{R} j^{-1}\left(\mathbb{J} j(\bvec{R}^{-1} \delta\bvec{R})\right)
        \end{pmatrix}\in T_q^*C.
    \end{equation}
\end{prop}
\begin{proof}
The first formula is simply: 
$$
\bvec{p} =\frac{\partial \mathcal{L}}{\partial \dot{\boldsymbol{\varphi}}} = \mathbb{A} \delta\boldsymbol{\varphi}.
$$
The second one results from the isometry $j$:
\begin{align*}
\ll \Sigma,\delta \bvec{R} \gg=
\ll \frac{\partial \mathcal{L}}{\partial \dot{\bvec{R}}},\delta \bvec{R} \gg
    &= <j(\bvec{R}^{-1} \delta \bvec{R}),\mathbb{J} j(\bvec{R}^{-1} \dot{\bvec{R}})>\\
    &= \ll \bvec{R}^{-1} \delta \bvec{R},j^{-1}\left(\mathbb{J} j(\bvec{R}^{-1} \dot{\bvec{R}})\right)\gg \\
    &= \ll \bvec{R}j^{-1}(\mathbb{J} j(\bvec{R}^{-1} \dot{\bvec{R}})),\delta \bvec{R} \gg . 
\end{align*}
where $\dot{\bvec{R}}=\frac{\partial \bvec{R}}{\partial t}$  and $\delta \bvec{R}$ is any element of the tangent space $T_{\bvec{R}} SO(3)^{[0,L]}$ of the group $SO(3)^{[0,L]}$ of rotations along the fiber $\mathfrak{C}$ at $\bvec{R}.$ 
\end{proof}

\subsection{Hamiltonian function and Poisson brackets (second formulation)}

The next proposition follows from a simple change of variables.

\begin{prop}
The Legendre transform of Prop.\ref{LegendreTransform} induces the Hamiltonian $H \colon T^*C \to \mathbb{R}$ defined as
\begin{equation}\label{eq:ham_pq}
\begin{split}
H(q,\pi) &=\frac{1}{2} \int_{0}^{L}
<\bvec{p} ,\mathbb{A}^{-1}\bvec{p} >
+
<j(\bvec{R}^{-1} \Sigma ),\mathbb{J}^{-1}j(\bvec{R}^{-1} \Sigma )> \\
&\quad\quad + <\frac{\partial \boldsymbol{\varphi}}{\partial S}-\bvec{d_3},\mathbb{G}(\frac{\partial \boldsymbol{\varphi}}{\partial S}-\bvec{d_3})>
+ <j(\bvec{R}^{-1} \frac{\partial \bvec{R}}{\partial S} ), \mathbb{H} j(\bvec{R}^{-1} \frac{\partial \bvec{R}}{\partial S} )> \, \dd S.
\end{split}
\end{equation}    
\end{prop}

\begin{rem}
    Using the metric of definition \ref{def:metrics}, the following fundamental duality relation (see chapter II, theorem 8.6 of \cite{Marle1987}) holds:
    \begin{equation}
        g\Bigl((\bvec{p} , \Sigma), (\delta \boldsymbol{\varphi}, \delta \bvec{R})\Bigr) = \mathcal{L}(\boldsymbol{\varphi},\bvec{R}, \delta\boldsymbol{\varphi}, \delta \bvec{R} ) + H(\boldsymbol{\varphi},\bvec{R},\bvec{p} , \Sigma).
    \end{equation}
\end{rem}

\begin{definition}
  On $T^*C = \{ (\boldsymbol{\varphi} , \bvec{R}, \bvec{p}, \Sigma)\},$ we define the canonical Poisson bracket:
\begin{equation}\label{eq:can_bracket}
    \{f,g\} = \int_0^L < \frac{\partial f}{\partial \boldsymbol{\varphi}}, \frac{\partial g}{\partial \bvec{p} } > - < \frac{\partial g}{\partial \boldsymbol{\varphi}}, \frac{\partial f}{\partial \bvec{p} } > + \ll \frac{\partial f}{\partial \bvec{R}}, \frac{\partial g}{\partial \Sigma} \gg - \ll \frac{\partial g}{\partial \bvec{R}}, \frac{\partial f}{\partial \Sigma} \gg \, \dd S.
\end{equation}    
\end{definition}

At this stage, let us make an important remark. Material indifference is phrased in variables $(\boldsymbol{\varphi}, \bvec{R})$ as invariance with respect to semidirect product $SO(3) \ltimes \mathbb{R}^3$ transformations: for any $\bvec{Q} \in SO(3)$ and $\bvec{c} \in \mathbb{R}^3,$ 
\begin{equation}\label{eq:matereial_ind}
    \begin{pmatrix}
        \boldsymbol{\varphi}\\
        \bvec{R}
    \end{pmatrix} \mapsto 
    \begin{pmatrix}
        \bvec{Q} \boldsymbol{\varphi} + \bvec{c}\\
        \bvec{Q} \bvec{R}
    \end{pmatrix}.
\end{equation}
With respect to the variables $(p , \boldsymbol{\sigma}, \boldsymbol{\varepsilon}, \boldsymbol{\kappa}),$ such a quasi-static rigid transformation gets reformulated as
\begin{equation}\label{eq:casimir_cause}
\begin{array}{ccl}
     \boldsymbol{\sigma} =& 0  &\text{ (quasi-static transformation)}\\
     \bvec{p}=& 0 &\text{ (quasi-static transformation)}\\     \boldsymbol{\kappa} =& 0   &\text{ (rigid transformation)}\\
     \boldsymbol{\varepsilon}= & 0 &\text{ (rigid transformation)}.
\end{array}
\end{equation}
In turn, these relations, when plugged in the Poisson bracket of theorem \ref{thm:old_ham_form}, give rise to the existence of non-trivial Casimir functions. More precisely, at the point given by equation \eqref{eq:casimir_cause}, the Poisson bracket of equations \eqref{eqn:main} becomes
\begin{align}
    \{f,g\}(0,0,0,0) = \int_0^L &  < \frac{\partial f}{\partial p }, \frac{\partial }{\partial S}\left( \frac{\partial g}{\partial \varepsilon} \right)> - < \frac{\partial g}{\partial p }, \frac{\partial }{\partial S}\left( \frac{\partial f}{\partial \varepsilon} \right)> \\
    +& < \frac{\partial f}{\partial \sigma}, \frac{\partial }{\partial S}\left( \frac{\partial g}{\partial \kappa} \right)> - < \frac{\partial g}{\partial \sigma}, \frac{\partial }{\partial S}\left( \frac{\partial f}{\partial \kappa} \right)>\dd S.
\end{align}
Hence, this bracket admits obvious degeneracies: any function $g$ admitting space-independent derivative will provide a Casimir function at zero, meaning that for any function $f$ of the variables $(p, \sigma, \varepsilon, \kappa),$
\begin{equation}
    \{f,g\}(0,0,0,0) = 0.
\end{equation}
The geometric reason for this is reduction theory of Poisson structures. By this, we mean that the Poisson bracket of theorem \ref{thm:old_ham_form} comes from a reduction procedure by quotienting out the Poisson structure \eqref{eq:can_bracket} through transformations of the form \eqref{eq:matereial_ind}. An interested reader can consult \cite{Marle1987}, chapter IV, for more insights on the topic.
One of the main ideas of the present work is the following: using material coordinates to construct Hamiltonian structure shortcuts this quotienting procedure because it ensures material invariance automatically. As we will see later on (see theorem \ref{thm:ham_form}), using the mobile frame to construct a Hamiltonian formulation of the Timoshenko model does not give rise to any non-trivial Casimir function.

\begin{prop}
In the coordinates $\{\boldsymbol{\varphi}, \bvec{p} , \bvec{R}, \Sigma\},$ the equations of motion are
\begin{equation}\label{eq:motionSIGMA}
\begin{alignedat}{6}
\frac{\partial \,\boldsymbol{\varphi}}{\partial t}&= \mathbb{A}^{-1}\bvec{p} 
&&\hspace{0.5cm}&
\frac{\partial \,\bvec{p} }{\partial t}&= \frac{\partial \, \mathbb{G}\boldsymbol{\varepsilon}}{\partial S}\\
\frac{\partial \,\bvec{R}}{\partial t}&= \bvec{R} j^{-1}(\mathbb{J}^{-1} j(\bvec{R}^{-1}\Sigma))
&&\hspace{0.5cm}&
\frac{\partial \, \Sigma}{\partial t } &= \bvec{R}(\frac{\partial\, \mathbb{H}\mathbb{K}}{\partial S}
 +
j^{-1}(\frac{\partial \boldsymbol{\varphi}}{\partial S}\wedge(\mathbb{G}\boldsymbol{\varepsilon})) + j^{-1}(\mathbb{J}^{-1}j(\bvec{R}^{-1} \Sigma))\bvec{R}^{-1}\Sigma).
\end{alignedat}
\end{equation}
\end{prop}

\subsection{Hamiltonian function and Poisson brackets (third formulation)}

In our task of formulating a Hamiltonian approach of the Timoshenko model, let us look for a better set of variables. $\Sigma$ -- being in the cotangent fiber above $\bvec{R}$ -- shall be replaced by the variable $\boldsymbol{\sigma}$ being its shifted counterpart above $Id.$ Indeed, the Hamiltonian \eqref{eq:ham_pq} and the Poisson brackets \eqref{eq:can_bracket} provide unsatisfactory Hamilton equations, them not being formulated on the mobile frame. In order to deal with this issue, we use -- in addition to Frobenius metric -- left translations on the group $SO(3)$ to identify the cotangent space $T^* SO(3)$ and $SO(3) \times so(3):$
\begin{equation}
\begin{array}{ccc}
    T^*SO(3)& \to &SO(3) \times so(3)  \\
    \Sigma(S) \in T^*_{\bvec{R}(S)} SO(3) & \mapsto & (\bvec{R}(S), \bvec{R}^{-1}(S) \Sigma(S).
\end{array}
\end{equation}

The map $j$ allows to represent any anti-symmetric matrix as a vector of $\mathbb{R}^3.$ Therefore, we obtain an isomorphism
\begin{equation}
\begin{array}{ccccccc}
    T^* C & \to & T^*\left( (\mathbb{R}^3)^{[0,L]}\right) &\times& SO(3)^{[0,L]} &\times& (\mathbb{R}^3)^{[0,L]}\\
    (\boldsymbol{\varphi}, \bvec{R}, \bvec{p} , \Sigma) & \mapsto & \boldsymbol{\varphi}, \bvec{p} && \bvec{R} && \boldsymbol{\sigma}
\end{array}
\end{equation}
where
\begin{equation}
    \boldsymbol{\sigma} = j(\bvec{R}^{-1} \Sigma)
\end{equation}
is the momentum of the rotation written on the moving frame using a $3$-dimensional vector.

\begin{rem}\label{rk:sigma_SO3}
    We emphasize an interesting feature of $\boldsymbol{\sigma}.$ By construction, $\boldsymbol{\sigma}$ is invariant by the action of $SO(3)$ on $T^* SO(3)$: for any $\bvec{Q} \in SO(3),$ and denoting $p_{\bvec{Q}\bvec{R}}$ the rotation momentum associated to $\bvec{Q} \cdot \bvec{R}(t),$ 
\begin{equation}
    j\left((\bvec{Q}\bvec{R})^{-1} (p_{\bvec{Q}\bvec{R}}) \right) = j(\bvec{R}^{-1} \bvec{Q}^{-1} \bvec{Q} \Sigma) =  \boldsymbol{\sigma}.
\end{equation}
This property of $\boldsymbol{\sigma}$ will be important in the construction of Hamilton's equations.
\end{rem}
Using the natural set of variables $(\boldsymbol{\varphi}, \bvec{p} , \bvec{R}, \boldsymbol{\sigma})$, it  is elementary to deduce the Hamiltonian from equation \eqref{eq:ham_pq}:
\begin{equation}\label{eq:triv_Ham}
            \begin{split}
            H(\boldsymbol{\varphi}, \bvec{p} , \bvec{R}, \boldsymbol{\sigma}) &=\frac{1}{2} \int_{0}^{L}
            <\bvec{p} ,\mathbb{A}^{-1}\bvec{p} >
            +
            <\boldsymbol{\sigma},\mathbb{J}^{-1}\boldsymbol{\sigma}> \\
            &\quad\quad + <(\frac{\partial \boldsymbol{\varphi}}{\partial S}-\bvec{d_3}),\mathbb{G}(\frac{\partial \boldsymbol{\varphi}}{\partial S}-\bvec{d_3})>
            + <j(\bvec{R}^{-1} \frac{\partial \bvec{R}}{\partial S} ), \mathbb{H} j(\bvec{R}^{-1} \frac{\partial \bvec{R}}{\partial S})> \, \dd S. 
            \end{split}
\end{equation}
We are now able to provide an alternative Poisson bracket.
\begin{theorem}
    In the coordinates $\{\boldsymbol{\varphi}, \bvec{p} , \bvec{R}, \boldsymbol{\sigma}\},$ the Poisson bracket \eqref{eq:can_bracket} becomes
    \begin{equation}\label{eq:triv_bracket}
    \begin{split}
        \{\Bar{f},\Bar{g}\} =& \int_0^L  < \frac{\partial \Bar{f}}{\partial \boldsymbol{\varphi}}, \frac{\partial \Bar{g}}{\partial \bvec{p} } > - < \frac{\partial \Bar{g}}{\partial \boldsymbol{\varphi}}, \frac{\partial \Bar{f}}{\partial \bvec{p} } >\\
        & + \ll \frac{\partial \Bar{f}}{\partial \bvec{R}}, \bvec{R}j^{-1}(\frac{\partial \Bar{g}}{\partial \boldsymbol{\sigma}}) \gg
        - \ll \frac{\partial \Bar{g}}{\partial \bvec{R}}, \bvec{R}j^{-1}(\frac{\partial \Bar{f}}{\partial \boldsymbol{\sigma}}) \gg \text{d}S. 
\end{split}
\end{equation}
\end{theorem}

\begin{proof}
Let $f \colon \{(\boldsymbol{\varphi},\bvec{p} , \bvec{R}, \Sigma) \} \to \mathbb{R}.$ In $(\boldsymbol{\varphi},\bvec{p} , \bvec{R}, \boldsymbol{\sigma})$ coordinates, 
\begin{equation}
    f(\boldsymbol{\varphi},\bvec{p} , \bvec{R}, \Sigma) = \Bar{f}(\boldsymbol{\varphi},\bvec{p} , \bvec{R}, j(\bvec{R}^{-1} \Sigma)) = \Bar{f}(\boldsymbol{\varphi},\bvec{p} , \bvec{R}, \boldsymbol{\sigma})
\end{equation} 
The equation:
\begin{equation}\label{eq:diff_pr}
    \frac{\partial f}{\partial  \Sigma} = \bvec{R} j^{-1} ( \frac{\partial \Bar{f}}{\partial \boldsymbol{\sigma}}).
\end{equation}
results from 

\begin{equation}
    \ll \frac{\partial f}{\partial \Sigma}, \delta \Sigma \gg = <\frac{\partial f}{\partial \boldsymbol{\sigma}}, j(\bvec{R}^{-1} \delta \Sigma) > = \ll \bvec{R} j^{-1} ( \frac{\partial \Bar{f}}{\partial \boldsymbol{\sigma}}), \delta \Sigma \gg.
\end{equation}
 
Since $\boldsymbol{\sigma}$ is $SO(3)$-invariant (see remark \ref{rk:sigma_SO3}):
\begin{equation}\label{eq:partialder_sigma}
    \ll \frac{\partial f}{\partial  \bvec{R}}, \delta \bvec{R} \gg = \ll \frac{\partial \Bar{f}}{\partial  \bvec{R}}, \delta \bvec{R} \gg.
\end{equation}
The Poisson bracket in the coordinates $\{\boldsymbol{\varphi}, p , \bvec{R}, \boldsymbol{\sigma}\}$ is therefore
\begin{align*}
    \{ \Bar{f}, \Bar{g} \} &= \int_0^L < \frac{\partial f}{\partial \boldsymbol{\varphi}}, \frac{\partial g}{\partial \bvec{p} } > - < \frac{\partial g}{\partial \boldsymbol{\varphi}}, \frac{\partial f}{\partial \bvec{p} } > + \ll \frac{\partial f}{\partial \bvec{R}}, \frac{\partial g}{\partial \Sigma} \gg - \ll \frac{\partial g}{\partial \bvec{R}}, \frac{\partial f}{\partial \Sigma} \gg \, \dd S\\
    &= \int_0^L < \frac{\partial \Bar{f}}{\partial \boldsymbol{\varphi}}, \frac{\partial \Bar{g}}{\partial \bvec{p} } > - < \frac{\partial \Bar{g}}{\partial \boldsymbol{\varphi}}, \frac{\partial \Bar{f}}{\partial \bvec{p} } >
    + \ll \frac{\partial \Bar{f}}{\partial  \bvec{R}}, \bvec{R} j^{-1} ( \frac{\partial \Bar{g}}{\partial \boldsymbol{\sigma}})  \gg - \ll \frac{\partial \Bar{g}}{\partial  \bvec{R}}, \bvec{R} j^{-1} ( \frac{\partial \Bar{f}}{\partial \boldsymbol{\sigma}})  \gg \, \dd S 
    \end{align*}
and the desired formula for the Poisson brackets is proved.
\end{proof}

\begin{rem}
Let us comment on the Poisson bracket described by equation \eqref{eq:triv_bracket}. The first line remains unchanged compared to the canonical bracket \eqref{eq:can_bracket}. The second line is similar, although subtler: the term $j^{-1}(\frac{\partial \Bar{g}}{\partial \boldsymbol{\sigma}})$ can be thought as being dual to an $S$-dependent rotation momentum, \textit{i.e.} an element of $T_{Id}\left(SO(3)^{[0,L]}\right) = so(3)^{[0,L]}.$ It is then left-translated by $\bvec{R}$ so that
\begin{equation}
    \bvec{R}j^{-1}(\frac{\partial \Bar{g}}{\partial \boldsymbol{\sigma}}) \in T_{\bvec{R}}\left(SO(3)^{[0,L]}\right).
\end{equation}
Since 
\begin{equation}
    \frac{\partial \Bar{f}}{\partial \bvec{R}} \in T^*_{\bvec{R}}\left(SO(3)^{[0,L]}\right),
\end{equation}
the term
\begin{equation}
    \ll \frac{\partial \Bar{f}}{\partial \bvec{R}}, \bvec{R}j^{-1}(\frac{\partial \Bar{g}}{\partial \boldsymbol{\sigma}}) \gg
\end{equation} is the natural pairing of a vector in $T\left(SO(3)^{[0,L]}\right)$ with a covector in $T^* \left(SO(3)^{[0,L]}\right).$
\end{rem}

Let us now write the dynamics in the coordinates we introduce. In addition to the definitions of momenta $\bvec{p}$ and $\boldsymbol{\sigma}$, equilibrium equations \eqref{eq:equillibrium} provide the following proposition.

\begin{prop}
In the coordinates $\{\boldsymbol{\varphi}, \bvec{p} , \bvec{R}, \boldsymbol{\sigma}\},$ the equations of motion are
\begin{equation}\label{eq:motion}
\begin{alignedat}{6}
\frac{\partial \,\boldsymbol{\varphi}}{\partial t}&= \mathbb{A}^{-1}\bvec{p} 
&&\quad\quad\quad\quad\hspace{1cm}&
\frac{\partial \,\bvec{p} }{\partial t}&= \frac{\partial \, \mathbb{G}\boldsymbol{\varepsilon}}{\partial S}\\
\frac{\partial \,\bvec{R}}{\partial t}&= \bvec{R} j^{-1}(\mathbb{J}^{-1} \boldsymbol{\sigma})
&&\quad\quad\quad\quad\hspace{1cm}&
\frac{\partial \, \boldsymbol{\sigma}}{\partial t } &= \frac{\partial\, \mathbb{H}\boldsymbol{\kappa}}{\partial S}
 +
\frac{\partial \boldsymbol{\varphi}}{\partial S}\wedge(\mathbb{G}\boldsymbol{\varepsilon}).
\end{alignedat}
\end{equation}
\end{prop}

\begin{rem}
\begin{equation}
    \boldsymbol{\sigma}=\mathbb{J}j(\bvec{R}^{-1}\frac{\partial \bvec{R}}{\partial t})    
\end{equation}
holds along a trajectory because in this equation, $\bvec{R}$ depends on $t.$ We recover therefore the definition
\begin{equation}
    \boldsymbol{\sigma}=\mathbb{J}\boldsymbol{\omega}
\end{equation}
of the kinetic momentum (also called \emph{moment of momentum}). In general, $\bvec{R}$ and $\boldsymbol{\sigma}$ are not related by any equation, because those are seen as two independent variables of the phase space. A similar remark holds for the displacement.
\end{rem}

We now combine the Hamiltonian \eqref{eq:triv_Ham} and the Poisson bracket \eqref{eq:triv_bracket} to obtain a new Hamiltonian formulation of Timoshenko model in Lagrangian coordinates.
%

\begin{theorem}[Hamiltonian formulation of Timoshenko model]\label{thm:ham_form}
Under suitable boundary conditions, the equations of motion \eqref{eq:motion} are Hamiltonian for the bracket \eqref{eq:triv_bracket} and the Hamiltonian \eqref{eq:triv_Ham}.
\end{theorem}

\begin{proof}
We are left to prove the four equations:
\begin{align}
    \{ \boldsymbol{\varphi}, H \} &= \mathbb{A}^{-1}\bvec{p} \label{eq:ham_phi}\\
\{\bvec{p} , H \} &= \frac{\partial \, \mathbb{G}\boldsymbol{\varepsilon}}{\partial S}\label{eq:ham_p_phi}\\
\{ \bvec{R}, H \}&= \bvec{R} j^{-1}(\mathbb{J}^{-1} \boldsymbol{\sigma})\label{eq:ham_R}\\
\{\boldsymbol{\sigma}, H \} &= \frac{\partial\, \mathbb{H}\boldsymbol{\kappa}}{\partial S}
 +
\frac{\partial \boldsymbol{\varphi}}{\partial S}\wedge(\mathbb{G}\boldsymbol{\varepsilon})\label{eq:ham_sigma}.
\end{align}
Applying formula \eqref{eq:triv_bracket}, 
\begin{align}
\{\boldsymbol{\varphi}, H \} &= \frac{\partial H}{\partial \bvec{p} } \\
\{\bvec{p} , H \} &= -\frac{\partial H}{\partial \boldsymbol{\varphi}}\\
\{\bvec{R}, H \}&= \bvec{R} j^{-1}(\frac{\partial H}{\partial \boldsymbol{\sigma}})\\
\{\boldsymbol{\sigma}, H \} &= - j(\bvec{R}^{-1}\frac{\partial H}{\partial \bvec{R}}).\label{la90}
\end{align}

From here, equations \eqref{eq:ham_phi} and \eqref{eq:ham_R} are straightforward. Equation \eqref{eq:ham_p_phi} is derived by integration by parts: under the boundary assumption 

\begin{equation}\label{eq:boundary1}
    <\delta \boldsymbol{\varphi}, \mathbb{G}\boldsymbol{\varepsilon} > = 0
\end{equation}
at $S = 0$ and $S=L$,
\begin{equation}
   \int_{0}^{L} <\frac{\partial H}{\partial \boldsymbol{\varphi}}, \delta \boldsymbol{\varphi}> \dd S = \int_{0}^{L} <\mathbb{G} \boldsymbol{\varepsilon}, \frac{\partial \delta\boldsymbol{\varphi}}{\partial S} > \, \dd S   = -\int_{0}^{L} <\frac{\partial \mathbb{G} \boldsymbol{\varepsilon} }{\partial S}   , \delta\boldsymbol{\varphi} > \, \dd S.
\end{equation}
In order to prove equation \eqref{eq:ham_sigma}, we compute $\frac{\partial H}{\partial \bvec{R}}.$ Let us do that in two steps. We first compute the derivative with respect to $\bvec{R}$ of 
\begin{equation}
    M(\bvec{R}) :=  \frac{1}{2} <\frac{\partial\boldsymbol{\varphi}}{\partial S} -\bvec{d_3},\mathbb{G}(\frac{\partial\boldsymbol{\varphi}}{\partial S}-\bvec{d_3})>.
\end{equation}
By noticing that (see also \eqref{eq:pertubation_eps}):
\begin{equation}
    \delta \boldsymbol{\varepsilon} - \delta \boldsymbol{\theta} \wedge \boldsymbol{\varepsilon} = \frac{\partial \, \delta \boldsymbol{\varphi} }{\partial S} - \delta \boldsymbol{\theta} \wedge \frac{\partial \boldsymbol{\varphi}}{\partial S},
\end{equation}
applying formula \eqref{PerturbHermitianProd} and the fact that $\frac{\partial \, \delta \boldsymbol{\varphi} }{\partial S}$ does not depend on $\bvec{R}$, we obtain

\begin{align}
    \ll \frac{\partial M}{\partial \bvec{R}}, \delta \bvec{R} \gg &= < - \delta \boldsymbol{\theta} \wedge \frac{\partial \boldsymbol{\varphi} }{\partial S}, \mathbb{G} \boldsymbol{\varepsilon} >\\
    &= - < \delta \boldsymbol{\theta}, \frac{\partial \boldsymbol{\varphi} }{\partial S} \wedge \mathbb{G} \boldsymbol{\varepsilon} >\\
    &= - \ll \bvec{R} j^{-1} (\frac{\partial \boldsymbol{\varphi} }{\partial S} \wedge \mathbb{G} \boldsymbol{\varepsilon}) , \delta \bvec{R} \gg
\end{align}\
and consequently
\begin{equation}
    \frac{\partial M}{\partial \bvec{R}} = - \bvec{R} j^{-1}(\frac{\partial \boldsymbol{\varphi} }{\partial S} \wedge \mathbb{G} \boldsymbol{\varepsilon}).
\end{equation}
We then compute the derivative with respect to $\bvec{R}$ of 
\begin{equation}
    F(\bvec{R}) := \frac{1}{2} \int_{0}^{L} <j(\bvec{R}^{-1} \frac{\partial\bvec{R}}{\partial S} ), \mathbb{H} j(\bvec{R}^{-1} \frac{\partial\bvec{R}}{\partial S} )> \, \dd S.
\end{equation}
We combine the perturbation of $\boldsymbol{\kappa}$ stated in equation \eqref{eq:var_omega_kappa} and the formula \eqref{PerturbHermitianProd}. Using the boundary conditions 
\begin{equation}\label{eq:boundary2}
    <\delta \boldsymbol{\theta}, \mathbb{H} \boldsymbol{\kappa}> = 0
\end{equation} at $S = 0,$ $S=L$ and any time,
\begin{align}\begin{split}
        \ll \frac{\partial F}{\partial \bvec{R}}, \delta \bvec{R} \gg &= \frac{1}{2}\int_{0}^{L} \delta < \boldsymbol{\kappa}, \mathbb{H} \boldsymbol{\kappa}> \dd S\\
    &= \int_{0}^{L} <\delta \boldsymbol{\kappa} - \delta \boldsymbol{\theta} \wedge \boldsymbol{\kappa}, \mathbb{H} \boldsymbol{\kappa}> \dd S\\
    &= \int_{0}^{L} <(\delta \boldsymbol{\theta})^{'}, \mathbb{H} \boldsymbol{\kappa}> \dd S\\
    &= - \int_{0}^{L} <\delta \boldsymbol{\theta}, \frac{\partial \mathbb{H} \boldsymbol{\kappa}}{\partial S} > \dd S\\
    &=- \int_{0}^{L} \ll \delta \bvec{R}, \bvec{R} j^{-1}(\frac{\partial \mathbb{H} \boldsymbol{\kappa}}{\partial S}) \gg \dd S.
\end{split}
\end{align}
As a consequence,
\begin{equation}
    \frac{\partial H}{\partial \bvec{R}} =     \frac{\partial M}{\partial \bvec{R}} +     \frac{\partial F}{\partial \bvec{R}} =  - \bvec{R} j^{-1}( \frac{\partial \boldsymbol{\varphi} }{\partial S} \wedge \mathbb{G} \boldsymbol{\varepsilon} ) - \bvec{R} j^{-1}(\frac{\partial \mathbb{H} \boldsymbol{\kappa}}{\partial S})
\end{equation}
and this last computation, with the help of \eqref{la90}, proves equation \eqref{eq:ham_sigma}.
\end{proof}

\begin{rem}[Relation in between boundary conditions and conservation of energy]
    In the proof, we used the boundary conditions \eqref{eq:boundary1} and \eqref{eq:boundary2}. They are similar to
    \begin{enumerate}
        \item the ones used in the variational principle of section \ref{sec:VariationalPrinciple1} to obtain the equations of equilibrium,
        \item the ones used to prove the conservation of energy in section \ref{sec:conservation}.
    \end{enumerate}
    Let us comment the second item. It is no surprise to observe a relation in between the conservative property of a mechanical system and the fact that it admits a Hamiltonian formulation. Nevertheless, there is a slight difference in between the two equations \eqref{eq:boundary1} and \eqref{eq:boundary2} and 
    \begin{equation}
    \Big[\bvec{v}\mathbb{G}\boldsymbol{\varepsilon}+\boldsymbol{\omega}\mathbb{H}\boldsymbol{\kappa}\Big]_{0}^{L} = 0.
    \end{equation}
    In both cases, a sufficient condition for these boundary terms to vanish is that $\boldsymbol{\varphi}$ and $\bvec{R}$ are constant with respect to time at extreme points of the fiber $\mathfrak{C}.$
\end{rem}

\subsection{Closure relations through the last Hamiltonian formulation}

Closure relations \eqref{crossderivation2} and \eqref{closurerelations2} were part of the set of equations \eqref{eq:comme_mar}. For the sake of completeness, we use the Hamiltonian formulation of theorem \ref{thm:ham_form} to recover the closure relations \eqref{crossderivation2} and \eqref{closurerelations2}. In that context, the right tool to handle Poisson brackets is test functions. Therefore, both closure relations are reformulated dually in terms of test functions in the sequel.

\subsubsection{Time-derivative of strain-vector}

\begin{prop}
    For any test function $f \colon (\mathbb{R}^3)^{[0,L]} \to \mathbb{R},$
    \begin{equation}
        \frac{\partial f}{\partial t}(\boldsymbol{\varepsilon}(t)) = \int_0^L < \frac{\partial f}{\partial \boldsymbol{\varepsilon}}, \frac{\partial \bvec{v}}{\partial S} - \boldsymbol{\omega}\wedge \bvec{d_3} > \text{d}S.
    \end{equation}
where $\bvec{v} = \mathbb{A}^{-1} \bvec{p}$ and $\boldsymbol{\omega} = \mathbb{J}^{-1}\boldsymbol{\sigma}$.
\end{prop}

\begin{proof}
Thanks to the theorem \ref{thm:ham_form}, it is enough to show:
\begin{equation}
    \{ f(\boldsymbol{\varepsilon}), H\} = \int_0^L < \frac{\partial f}{\partial \boldsymbol{\varepsilon}}, \frac{\partial \bvec{v}}{\partial S} - \boldsymbol{\omega}\wedge \bvec{d_3} > \text{d}S.
\end{equation}
Using the Poisson bracket \eqref{eq:triv_bracket},
\begin{equation}
    \{ f(\boldsymbol{\varepsilon}), H\} = \int_0^L < \frac{\partial f}{\partial \boldsymbol{\varphi}}, \mathbb{A}^{-1} \bvec{p}> + \ll \frac{\partial f}{\partial \bvec{R}}, \bvec{R} j^{-1}(\boldsymbol{\omega}) \gg \text{d}S.
\end{equation}
In the previous sum, the first term is computed by the chain rule:
\begin{align}\label{eq:eps_first_term}
     < \frac{\partial f}{\partial \boldsymbol{\varphi}}, \mathbb{A}^{-1} \bvec{p}> &= < \frac{\partial f}{\partial \boldsymbol{\varepsilon}}, \frac{\partial \boldsymbol{\varepsilon}}{\partial \boldsymbol{\varphi}} \mathbb{A}^{-1} \bvec{p}>\\
     &= < \frac{\partial f}{\partial \boldsymbol{\varepsilon}}, \frac{\partial  \mathbb{A}^{-1} p}{\partial S}>.
\end{align}
The second term
\begin{equation}\label{eq:eps_second_term}
    \ll \frac{\partial f}{\partial \bvec{R}}, \bvec{R} j^{-1}(\boldsymbol{\omega}) \gg = <\frac{\partial f}{\partial \boldsymbol{\varepsilon}}, \frac{\partial \boldsymbol{\varepsilon}}{\partial \bvec{R}} \bvec{R} j^{-1}(\boldsymbol{\omega})>
\end{equation}
requires to compute $\frac{\partial \boldsymbol{\varepsilon}}{\partial \bvec{R}} \colon T_{\bvec{R}}\left( SO(3) \right)^{[0,L]} \to \left( \mathbb{R}^3 \right)^{[0,L]}.$ Any element of $T_{\bvec{R}}\left( SO(3) \right)^{[0,L]}$ can be written as $\bvec{R Q},$ with $\bvec{Q} \in \left( so(3) \right)^{[0,L]}$. So, we use equation \eqref{eq:pertubation_eps} to compute $\frac{\partial \boldsymbol{\varepsilon}}{\partial \bvec{R}}\cdot \bvec{RQ}$ for any $\bvec{Q} \in \left( so(3) \right)^{[0,L]}$:
\begin{equation}
    \frac{\partial \boldsymbol{\varepsilon}}{\partial \bvec{R}}\cdot \bvec{R} \bvec{Q} = j(\bvec{Q}) \wedge \bvec{d_3}.
\end{equation}
To conclude, we set $\bvec{Q} = j^{-1}(\boldsymbol{\omega})$ in the previous equation and gather \eqref{eq:eps_first_term} and \eqref{eq:eps_second_term} to provide
\begin{equation}
    \frac{\partial f(\boldsymbol{\varepsilon})}{\partial t} = \int_0^L < \frac{\partial f}{\partial \boldsymbol{\varepsilon}}, \frac{\partial \mathbb{A}^{-1} \bvec{p}}{\partial S} - \boldsymbol{\omega}\wedge \bvec{d_3} > \text{d}S.
\end{equation}
\end{proof}
\subsubsection{Time-derivative of spatial curvature}
\begin{prop}
    For any test function $f \colon (\mathbb{R}^3)^{[0,L]} \to \mathbb{R},$
    \begin{equation}
        \frac{\partial f}{\partial t}(\boldsymbol{\kappa}(t)) = \int_0^L < \frac{\partial f}{\partial \boldsymbol{\kappa} }, \frac{\partial \boldsymbol{\omega}}{\partial S} + \boldsymbol{\omega}\wedge \boldsymbol{\kappa} > \text{d}S.
    \end{equation}
\end{prop}
\begin{proof}
Thanks to the theorem \ref{thm:ham_form}, it is enough to show:
\begin{equation}
    \{ f(\boldsymbol{\kappa}), H\} = \int_0^L < \frac{\partial f}{\partial \boldsymbol{\kappa}}, \frac{\partial \boldsymbol{\omega}}{\partial S} + \boldsymbol{\omega} \wedge \boldsymbol{\kappa} > \text{d}S.
\end{equation}
Using the Poisson bracket \eqref{eq:triv_bracket},
\begin{align}
    \{ f(\boldsymbol{\kappa}), H\} &= \int_0^L \ll \frac{\partial f}{\partial \bvec{R}} , \bvec{R} j^{-1}(\boldsymbol{\omega}) \ll \text{d}S\\
    &= \int_0^L < \frac{\partial f}{\partial \boldsymbol{\kappa}}, \frac{\partial \boldsymbol{\kappa}}{\partial \bvec{R}} \bvec{R} j^{-1}(\boldsymbol{\omega})> \text{d}S.
\end{align}
We are left to show: 
\begin{equation}\label{eq:derivative_kappa}
    \frac{\partial \boldsymbol{\kappa}}{\partial \bvec{R}} \cdot \left( \bvec{R} j^{-1}(\boldsymbol{\omega})\right) = \frac{\partial \boldsymbol{\omega}}{\partial S} + \boldsymbol{\omega} \wedge \boldsymbol{\kappa}.
\end{equation}
Let us recall
$    \mathbb{K} = \bvec{R}^{-1} \frac{\partial \bvec{R}}{\partial S}.
$ From \cite{marsden2013introduction}, section 13.5, for any $\bvec{Q} \in so(3)^{[0,L]}$:
\begin{equation}
    \frac{\partial \mathbb{K}}{\partial \bvec{R}} \cdot (\bvec{R} \bvec{Q}) = \bvec{R}^{-1} \frac{\partial \bvec{R} \bvec{Q}}{\partial S} - \bvec{R}^{-1} \frac{\partial \bvec{R} }{\partial S}\bvec{Q}
\end{equation}
and
\begin{equation}
    \frac{\partial \bvec{Q}}{\partial S}
     = \bvec{R}^{-1} \frac{\partial \bvec{R} \bvec{Q}}{\partial S}- \bvec{Q} \bvec{R}^{-1} \frac{\partial \bvec{R}}{\partial S}.
\end{equation}
Combining the two last equations (see also equation \eqref{eq:cocycle_space}):
\begin{equation}
    \frac{\partial \mathbb{K}}{\partial \bvec{R}} \cdot \left( \bvec{R} \bvec{Q} \right) = \frac{\partial \bvec{Q}}{\partial S}
     + \left[\bvec{Q}, \bvec{R}^{-1} \frac{\partial \bvec{R}}{\partial S}\right].
\end{equation}
Since $\boldsymbol{\kappa} = j(\mathbb{K}),$\begin{equation}\label{eq:derivative_kappa2}
\frac{\partial \boldsymbol{\kappa}}{\partial \bvec{R}} \cdot \left(\bvec{R} \bvec{Q} \right)= j(\frac{\partial  \bvec{Q}}{\partial S}) + j(\bvec{Q}) \wedge \boldsymbol{\kappa}.
\end{equation}
We conclude at equation \eqref{eq:derivative_kappa} by setting $\bvec{Q} = j^{-1}(\boldsymbol{\omega})$ in equation \eqref{eq:derivative_kappa2}.
\end{proof}
\section{Conclusion and perspectives}

This paper has highlighted certain features that are intrinsic to the modeling of beams. First, such one-dimensional body, lying in a three-dimensional ambient space, could be modeled in various way: privileging the ambient space or the material space. This is particularly true for large transformations where Lagrangian and Eulerian formulation stand out. Second, these Cosserat structures have (at each point) degrees of freedom of different nature: translation and rotation. Accordingly, it is not surprising to observe that the Hamiltonian structure of such material body could be presented rigorously in various manner, all preserving the geometrical and material background, but focusing on various variables. \\
Three Hamiltonian structures and associated Poisson brackets have been presented (respectively \eqref{eq:Hamilton1} and \eqref{eqn:main}, \eqref{eq:ham_pq} and \eqref{eq:can_bracket}, \eqref{eq:triv_Ham} and \eqref{eq:triv_bracket}). The first one focuses on variables associated to the tangent space (secondary variable) and lead to dynamical equations written as two sets of first-order differential equations. These secondary variables have first-order shape as they are associated to vectors, co-vectors, or elements of $so(3)$. The two others are explicitly written on the mobile frame.  The second one is the canonical Hamiltonian structure of classical mechanics. The last one is formulated thanks to first-order degrees of freedom, \textit{i.e.} translations and rotations of the beam section.\\
This paper focuses on these Hamiltonian structures for large dynamical transformations in three-dimensional ambient space. As a consequence, less attention has been paid to external loading and non-standard boundary conditions (\textit{e.g.} Robin-type or non-holonomic contact). This technical aspect is postponed to future work dealing with more practical applications. Another extension of this work could be to extend derivatives and perturbations of quadratic forms in order to deal with a non-quadratic strain energy, typically a non-linear stress-strain relationship. However, the problem presented in this paper is sufficiently general to model dynamics of rigid bodies, linear vibrations \cite{reddy1999dynamic}, large transformations of string, buckling of beam \cite{spagnuolo2019targeted}, or three-dimensional dynamics. \\
Beyond theoretical interest, it seems to us that the exhibit of Hamiltonian structures may be of interest from a numerical point of view. In fact, each formulation can be used to construct a numerical integration scheme that preserves the Hamiltonian structure of the problem. Moreover, having three distinct formulations can then be useful so that the integration of each variable is subject to this preservation. \\
Another application of these Hamiltonian formulations concerns the loss of symmetry during a bifurcation. Indeed, it could be interesting to observe how each formulation is sensitive to such instabilities and how it expresses these symmetries and their loss. As an example, one could show that the plane problem in the $(\bvec{e_1},\bvec{e_3})$-plane is a sub-Hamiltonian system of the Hamiltonian system \eqref{eq:motion} in the following sense: first, the phase space of the plane problem is a Poisson sub-manifold of the space $\{(\boldsymbol{\varphi}, \bvec{p}, \bvec{R}, \boldsymbol{\sigma})\},$ meaning that the restriction of the variables $(\boldsymbol{\varphi}, \bvec{p} , \bvec{R}, \boldsymbol{\sigma})$ to the plane problem still provides a Poisson bracket, say $\{.,.\}_{\text{sub}},$ on this smaller space. Second, the plane problem in the $(\bvec{e_1},\bvec{e_3})$-plane is Hamiltonian for the restriction of the Hamiltonian \eqref{eq:triv_Ham} and the Poisson bracket $\{.,.\}_{\text{sub}}.$ Hamiltonian formulations of Timoshenko model could be used to study its instabilities by reformulating them geometrically. We are then left with the study of a Hamiltonian system in a neighborhood of a Poisson submanifold. This geometric approach to beam instabilities could be enriched by the numerous Hamiltonian formulations listed in the paper. 
This investigation is beyond the scope of this paper and postponed to future work.

\newpage

\bibliography{biblio}

\end{document}